%% file: main.tex
\documentclass[lettersize,journal]{IEEEtran}
\usepackage{amsmath,amsfonts}
\usepackage{algorithmic}
\usepackage{algorithm}
\usepackage{array}
\usepackage[caption=false,font=normalsize,labelfont=sf,textfont=sf]{subfig}
\usepackage{textcomp}
\usepackage{stfloats}
\usepackage{url}
\usepackage{verbatim}
\usepackage{graphicx}
\usepackage{cite}
\usepackage{hyperref}
\hypersetup{colorlinks,urlcolor=blue}
\input{commands}

\usepackage{lipsum}
\usepackage{diagbox}
\usepackage{multirow}
\usepackage{pifont}
\usepackage[normalem]{ulem}
\useunder{\uline}{\ul}{}
\usepackage{booktabs}
\usepackage[table]{xcolor}

\definecolor{weLow}{rgb}{0.8,0.2,0.2}
\definecolor{weMid}{rgb}{0.95,0.95,0.95}
\definecolor{weHigh}{rgb}{0.0,0.8,0.6}

\newcommand{\wecolorpm}[2]{%
  \ifnum\fpeval{#1<0.5}>0
    \cellcolor{weLow!\fpeval{100-#1*200}!weMid}#1 $\pm$ #2%
  \else
    \cellcolor{weMid!\fpeval{100-(#1-0.5)*200}!weHigh}#1 $\pm$ #2%
  \fi
}

\newcommand{\prcolorpm}[2]{%
  \ifnum\fpeval{#1<0.5}>0
    \cellcolor{weLow!\fpeval{100-#1*200}!weMid}#1 $\pm$ #2%
  \else
    \cellcolor{weMid!\fpeval{100-(#1-0.5)*200}!weHigh}#1 $\pm$ #2%
  \fi
}

\begin{document}

\title{Wavefront Estimation From a Single Measurement: Uniqueness and Algorithms
}

\author{Nicholas Chimitt,~\IEEEmembership{Member,~IEEE,} Ali Almuallem,~\IEEEmembership{Student Member,~IEEE,} \\
Qi Guo,~\IEEEmembership{Member,~IEEE,} Stanley H. Chan,~\IEEEmembership{Senior Member,~IEEE}
\thanks{The authors are with the School of Electrical and Computer Engineering, Purdue University, West Lafayette, IN 47907, USA. Corresponding author: Nicholas Chimitt, email: \texttt{nchimitt@purdue.edu}.}
\thanks{The work is supported in part by the Intelligence Advanced Research Projects Activity (IARPA) under Contract No. 2022-21102100004, and in part by the National Science Foundation under the grants CCSS-2030570 and IIS-2133032. The views and conclusions contained herein are those of the authors and should not be interpreted as necessarily representing the official policies, either expressed or implied, of IARPA, or the U.S. Government. The U.S. Government is authorized to reproduce and distribute reprints for governmental purposes notwithstanding any copyright annotation therein.}
}

\maketitle

\begin{abstract}
Wavefront estimation is an essential component of adaptive optics where the goal is to recover the underlying phase from its Fourier magnitude. While this may sound identical to classical phase retrieval, wavefront estimation faces more strict requirements regarding uniqueness as adaptive optics systems need a unique phase to compensate for the distorted wavefront. Existing real-time wavefront estimation methodologies are dominated by sensing via specialized optical hardware due to their high speed, but they often have a low spatial resolution. A computational method that can perform both fast and accurate wavefront estimation with a single measurement can improve resolution and bring new applications such as real-time \emph{passive} wavefront estimation, opening the door to a new generation of medical and defense applications.

In this paper, we tackle the wavefront estimation problem by observing that the non-uniqueness is related to the geometry of the pupil shape. By analyzing the source of ambiguities and breaking the symmetry, we present a joint optics-algorithm approach by co-designing the shape of the pupil and the reconstruction neural network. Using our proposed lightweight neural network, we demonstrate wavefront estimation of a phase of size $128\times 128$ at $5,200$ frames per second on a CPU computer, achieving an average Strehl ratio up to $0.98$ in the noiseless case. We additionally test our method on real measurements using a spatial light modulator. Code is available at \url{https://pages.github.itap.purdue.edu/StanleyChanGroup/wavefront-estimation/}.
\end{abstract}

\begin{IEEEkeywords}
Wavefront estimation, wavefront sensing, adaptive optics, phase retrieval, neural representation
\end{IEEEkeywords}

\section{Introduction}
\IEEEPARstart{I}{maging} systems are impacted by aberrations when a distorted wavefront is projected to a sensor, creating a less-than-ideal impulse response known as a point spread function (PSF). While image recovery can be done via post-processing, a commonly used hardware solution is adaptive optics (AO) which aims to cancel out the wavefront distortion in real-time and bring the system to its diffraction limit. The AO system has two tasks: wavefront estimation and correction. Ideally, we want to do wavefront estimation fast; if the distortion fluctuates in time, latency in estimation will reduce our ability to mitigate the effects due to phase decorrelation \cite{Tyson_Frazier_Book}.

This paper focuses on the estimation problem, i.e., recovering the wavefront from the PSF. We focus on the hardest limit: using only one \textit{single} measurement. At the core of this problem is the recovery of the phase from a Fourier magnitude. This setting should be familiar to many readers who know phase retrieval. Indeed, wavefront estimation \emph{is} a form of phase retrieval with a subtle but important difference. Both problems can be formulated as the recovery of a complex signal from its Fourier magnitude: if $\vx \in \C^N$ is the ground truth complex signal and $\mF \in \C^{M \times N}$ is an oversampled discrete Fourier transform matrix, then the inverse problem is to find $\vx$ from the measurement $\vy \in \R^M$ \cite{sayre_1952_a, Gerchberg_1972, Fienup_1982_Comparison, Eldar_2016_Hassibi_overview, bendory_2017_a, Unser_2023_Review}:
\begin{equation}
\text{Find} \;\; \vx \quad \text{subject to} \quad \vy = |\mF\vx|^2.
\label{eq: simple phase retrieval}
\end{equation}
However, because of the magnitude-square operation, the problem does not always have a unique solution.

\begin{figure*}[t]
\centering
\includegraphics[width=0.9\linewidth, trim={0 0 0 1cm},clip]{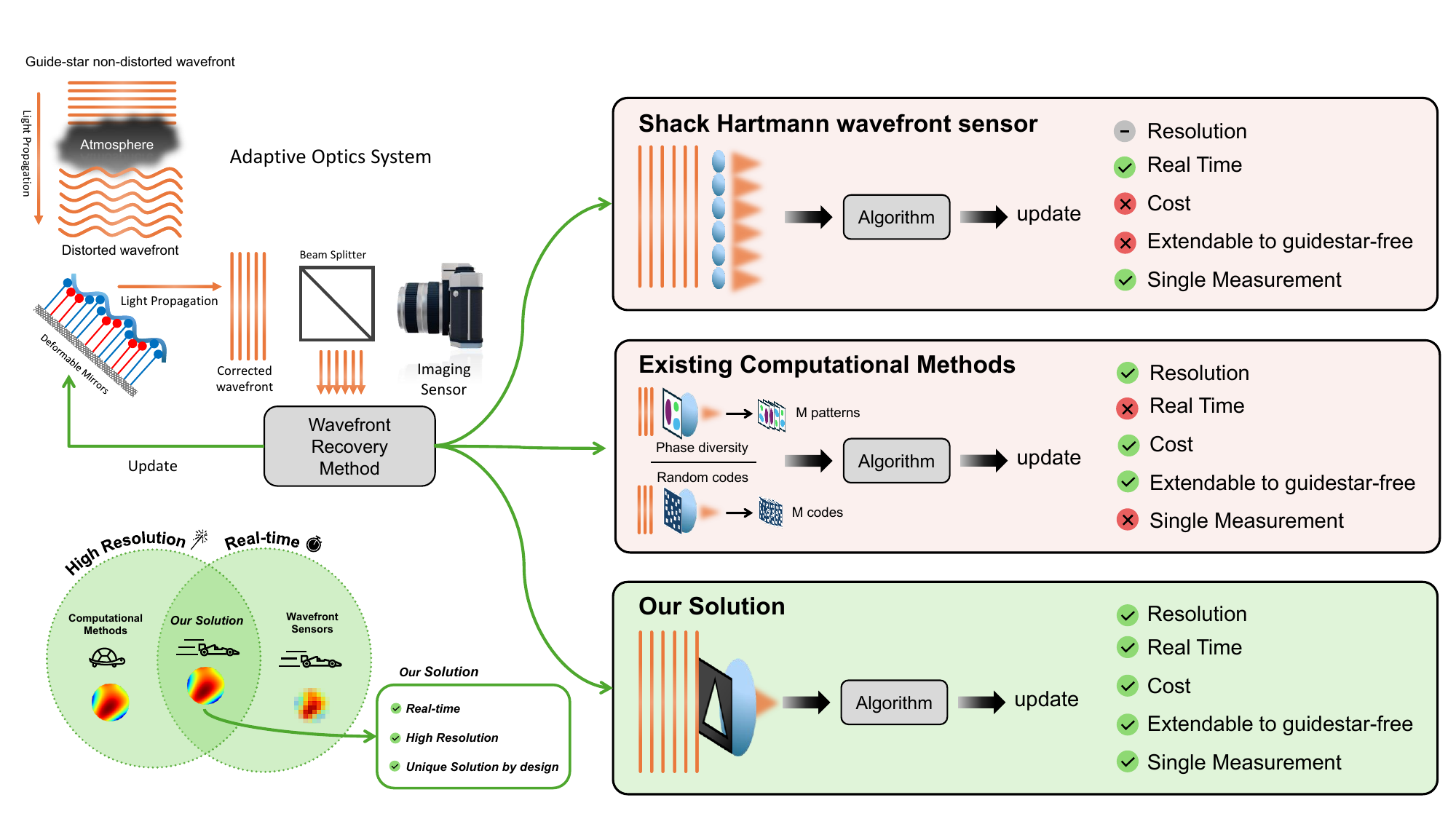} \vspace{-3ex}
\caption{Overview of adaptive optics (AO) and the position of our method among hardware and alternative computational methods for wavefront estimation.}
\label{fig: glass}
\end{figure*}

The subtle but important difference between phase retrieval and wavefront estimation lies in the types of permissible ambiguous solutions of \eqref{eq: simple phase retrieval}, divided into \emph{trivial} and \emph{non-trivial} ambiguities. As the names imply, trivial ambiguities are viewed as acceptable solutions to \eqref{eq: simple phase retrieval} in the context of phase retrieval \cite{Ranieri_2013, Eldar_2014_STFT, bendory_2017_a, Eldar_2016_Hassibi_overview, Unser_2023_Review}. The vast majority of phase retrieval papers is on eliminating non-trivial ambiguities.

The story is different for wavefront estimation. While wavefront estimation still suffers from the non-trivial ambiguity problem, there also exist trivial ambiguities that will severely degrade AO correction. Because of this, wavefront estimation is often performed by specialized optical hardware or, if done computationally, using multiple measurements \cite{gonsalves_1982_a, Candes_2015_Coded} to overcome the relevant trivial ambiguities.

\subsection{Why not just measure it?} 
If dealing with the trivial ambiguity is such an important but difficult problem, why don't we just \emph{measure} the phase using a Shack-Hartmann wavefront sensor \cite{Ragazzoni_1996, Tyson_Frazier_Book}? Putting aside the cost (a typical setup from Thorlabs would cost thousands of dollars), the resolution is in the range of $35 \times 21$ to $73 \times 45$. This means it cannot measure too many high-order aberrations, in addition it has been shown a Shack-Hartmann sensor can suffer when the aberrations are strong due to an effect analogous to crosstalk. Alternatives such as interferometry \cite{Hariharan_1987_PhaseShift, Colavita_1994_Interferometry, Wyant_2003_Interferometry} can measure high-order aberrations, but are more suitable for laboratory settings due to their reliance on precise calibration.

Computational methods bring the possibility of high-performance, guidestar-free adaptive optics. Recently, Feng et al. \cite{Feng_2023_NeuWS} demonstrated that from a natural blurry image, one can recover the PSF and corresponding wavefront with a compute time on the order of minutes. This paper offers the complementary capability of recovering the wavefront from a single PSF at thousands of frames per second. We provide an overview of this method in the context of alternatives in \fref{fig: glass}.

\subsection{Preview of main results}
A key observation we make in this paper is that most pupils are a circle --- this is how we design our imaging systems. However, a circular pupil creates a critical issue in wavefront estimation: the support of its conjugate-flip is identical. So, if we want to resolve trivial ambiguity, we need to \emph{break the symmetry}! As we will show later in the paper, we can identify conditions under which there exist \emph{invertible point spread functions (iPSFs)} such that the phase can be recovered suitably for adaptive optics.

\begin{figure}[t]
\centering
\begin{tabular}{c}
\includegraphics[width=0.9\linewidth]{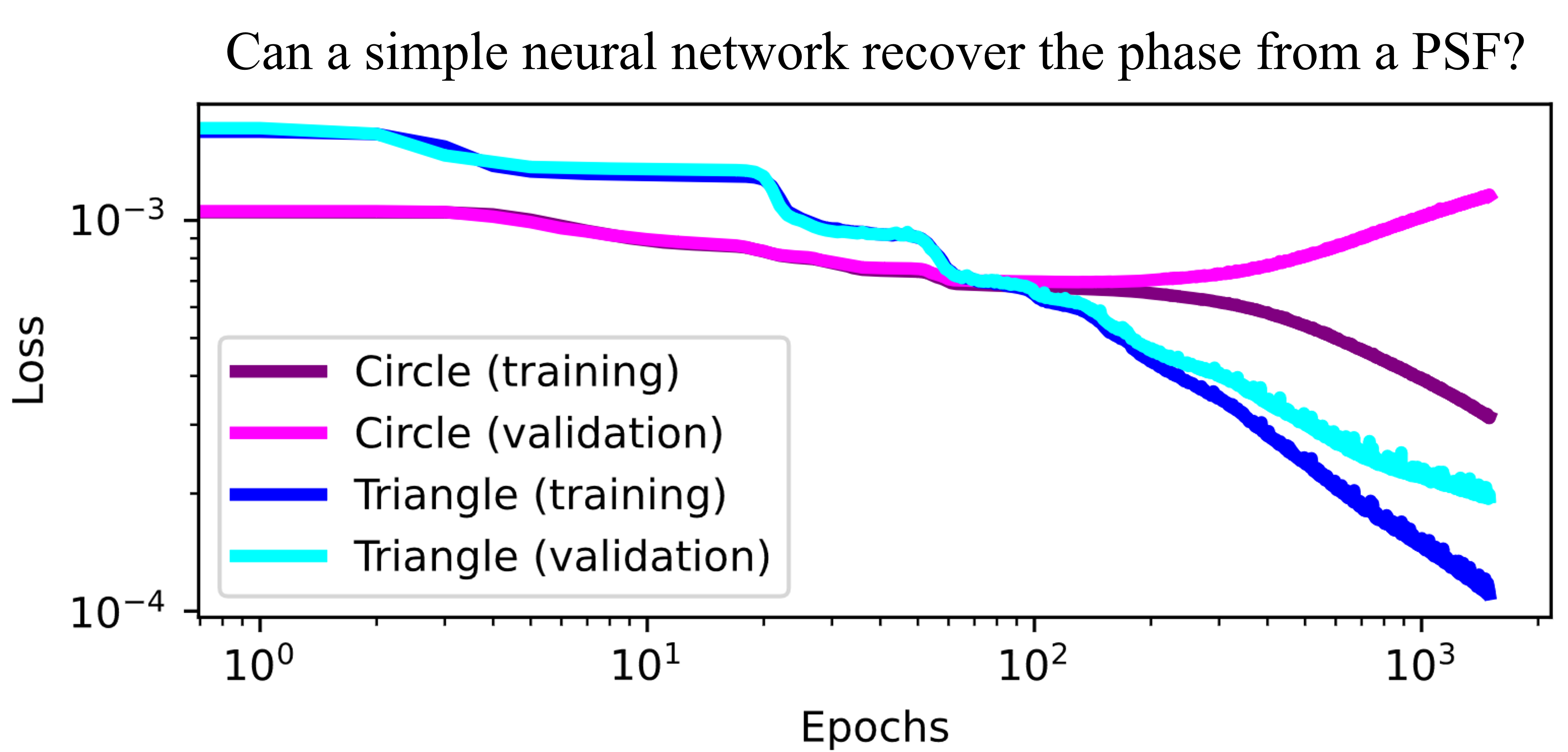}
\end{tabular}
\vspace{-2ex}
\caption{
Training and validation loss for a dataset of PSFs observed through circular and triangular pupils. Using a circular pupil causes overfitting while the triangular pupil allows reasonable generalization to validation data.}
\label{fig: s2p_doesnt_work 1}
\end{figure}

To convince readers of the significance of pupil geometry, in \fref{fig: s2p_doesnt_work 1} we show the training and validation performance for recovering phase uniquely from a Fourier magnitude. The ``algorithm'' we use to perform the phase recovery task is to train a simple multi-layer perceptron (MLP); we are not advocating in favor of a learning-based solution over other well-established convex optimization algorithms, we are merely illustrating the importance of geometry. The training and validation curves in \fref{fig: s2p_doesnt_work 1} show a clear contrast between the two geometries. For the circular aperture, even if we can drive the training loss down, the validation loss will increase, implying that the network cannot generalize. The triangle aperture has a much better behavior as far as training and validation are concerned. As we shall explain more precisely later, this is not a failure of network design -- it is explainable from the types of ambiguities permitted by the two geometries.

We outline the contributions of this paper as follows:
\begin{enumerate}
    \item \textbf{Practical conditions for invertible PSFs.} We prove there exists a set of practically achievable conditions for wavefront estimation from an incoherent PSF in which the solution is unique. Therefore, one can build an imaging system whose phase can be uniquely determined. We refer to these as invertible PSFs. 
    \item \textbf{Unique wavefront estimation algorithms using a single PSF.} We describe and compare a few varieties of networks that can perform real-time PSF inversion, operating near or faster than consumer-grade high-speed Shack-Hartmann sensors and offering high accuracy for downstream adaptive optics correction.
\end{enumerate}

\section{Background and Phase Retrieval}
\subsection{Notation}
In this paper, we focus on 2D problems, though it can be generalized to $d$ dimensions \cite{Hayes_1982_Multidimensional, Ranieri_2013, bendory_2017_a}. A 2D complex signal is denoted as $x[\vn] \in \C$, where $\vn = (n_1,n_2)$ denotes the spatial index. Without loss of generality, we assume $\vn$ is defined on a regular grid so that $n_i \in \{0,\ldots,N_i-1\}$ where $N_i$ is the number of points along each dimension of the grid $\Z_N = \{0, \ldots, N_1-1\} \times \{0, \ldots, N_2-1\}$.

The signal $x[\vn]$ contains magnitude and phase $\phi[\vn] \in \R$,
\begin{equation}
x[\vn] = \big| x[\vn] \big| \; e^{j \phi[\vn]}.
\end{equation}
We use $\vx \in \C^N$ to denote an $N$-dimensional vector $\vx = [x[\vn_1],\ldots,x[\vn_N]]^T$. For a 2D regular grid, it is easy to see that $N = N_1N_2$. To specify the vector of the magnitude of $\vx$, we write $|\vx| = [|x[\vn_1]|,\ldots,|x[\vn_N]|]^T$. Similarly, for the phase, we write $\vphi = [\phi[\vn_1],\ldots,\phi[\vn_N]]^T$.

We define the discrete Fourier transform (DFT) as
\begin{equation}
y[\vm] =
\frac{1}{M} \sum_{\vn \in \Z_N} x[\vn] e^{-j2\pi \vm^T\vn / M },
\label{eq: measurement_model}
\end{equation}
where $\vm = (m_1,m_2)$ is the output coordinate with $m_i = 0,\ldots,M_i-1$, and $M = M_1M_2$. The output vector is defined as $\vy = [y[\vm_1],\ldots,y[\vm_M]]^T$.

For a typical DFT, the input dimension $N_i$ and output dimension $M_i$ satisfy $M_i = N_i$. However, in phase retrieval, we often need to \emph{oversample} in the output space. In this case, we will have $M_i \ge N_i$, corresponding to padding $M_i - N_i$ zeros to each dimension of the input $x[\vn]$ before performing the DFT. Assuming that the oversampling is equal across all dimensions, we define the oversampling ratio as $s \bydef M/N$. For notational convenience, we also define the oversampled Fourier transform matrix $\mF \in \C^{M \times N}$ with $M \ge N$. Therefore, $\vy = \mF\vx$ represents a $M\times N$ discrete Fourier transform on a zero-padded input $x[\vn]$.

\subsection{Phase retrieval and trivial ambiguities}
\begin{definition}[Phase Retrieval]
Let $\vx \in \C^N$ be the ground truth complex signal, and let $\vy = |\mF\vx|^2 \in \R^M$ be the Fourier magnitude-square where $\mF \in \C^{M \times N}$ is the oversampled DFT matrix.
The phase retrieval problem is defined as
\begin{equation}
\text{Find} \quad \widehat{\vx} \qquad \text{such that} \quad \vy = \abs{\mF \widehat{\vx}}^2.
\label{eq: pr_matrix}
\end{equation}
\end{definition}

In this paper, we are interested in the following three types of solutions that can be derived from the true signal $\vx$.
\begin{definition}
Let $x[\vn]$ be the ground truth signal. We define DC phase offset, translation\footnote{The translation here is a circular translation. This ensures consistency when we perform Fourier analysis. Thus, strictly speaking, we should write $\calT_{\text{shift}}\{x[\vn]\} = x[\vn - \vell \;\; \text{mod} \; N]$. But given this cumbersome notation, we write $x[\vn - \vell]$ with the implied meaning of circular wrapping.}, and conjugate flip as follows.
\begin{alignat*}{2}
\calT_{\text{dc}}\{\vx\}[\vn] &= x[\vn]e^{j\phi_0}, &\qquad  & \text{for any } \phi_0 \in [0,2\pi), \\
\calT_{\text{shift}}\{\vx\}[\vn] &= x[\vn - \vell],     &\qquad  & \text{for any } \vell \in \Z_N,  \\
\calT_{\text{flip}}\{\vx\}[\vn] &= x^*[-\vn].
\end{alignat*}
These operations can be defined elementwise for any vector $\vx$, thus giving us $\calT_{\text{dc}}(\vx)$, $\calT_{\text{shift}}(\vx)$, $\calT_{\text{flip}}(\vx)$, respectively.
\end{definition}

\begin{definition}[Trivial Ambiguity]
Given a vector $\vx$, the trivial ambiguity set $\Omega(\vx)$ contains all vectors generated by
\begin{equation}
\Omega(\vx) = \langle \{ \calT_{\text{dc}}(\vx), \; \calT_{\text{shift}}(\vx), \; \calT_{\text{flip}}(\vx) \} \rangle,
\label{eq: ambiguities}
\end{equation}
where $\langle \{ \} \rangle$ denotes closure of the set under compositions.
\end{definition}

\begin{property}[Bendory et al. \cite{bendory_2017_a}] Let $\vx \in \C^N$ be a complex vector and $\mF \in \C^{M \times N}$ be the Fourier transform matrix. Then, any trivial ambiguity will give us the same Fourier magnitude:
\begin{equation}
\widehat{\vx} \in \Omega(\vx) \; \implies \; \abs{\mF \widehat{\vx}}^2 = \abs{\mF \vx}^2.
\end{equation}
\end{property}
In addition to trivial ambiguities, there also exist non-trivial ambiguities such that $\widehat{\vx} \not\in \Omega(\vx)$ but $\abs{\mF \widehat{\vx}}^2 = \abs{\mF \vx}^2$. Non-trivial ambiguities are of significant importance in the literature as they correspond to signals with an entirely different structure than what $\vx$ actually represents.

\subsection{Known results about uniqueness}
In the phase retreival literature, an algorithm declares success if the only ambiguities its solution encounters are the trivial ambiguities \cite{Fienup_1983_Reference, Fienup_1983_Uniqueness}. Most of the known uniqueness results are the non-trivial ones. These non-trivial solutions are plenty for 1D problems, e.g., Hassibi and collaborators \cite{Eldar_2016_Hassibi_overview} mentioned that there are up to $2^N$ of them where $N$ is the number of variables (i.e., pixels) for the phase.

Overcoming the non-trivial ambiguities is usually centered around two themes: oversampling rate and prior information (e.g., sparsity) \cite{Eldar_2016_Hassibi_overview}. The minimum oversampling rate is $M \ge 2N$ if the sampling matrix is the DFT matrix. For a general sampling matrix, the consensus is $M \ge 4N - 4$ \cite{Unser_2023_Review} which is called the ``$(4N-4)$ conjecture'' \cite{Bandeira_2014_SavingPhase} and is proven in \cite{Conca_2015_Injectivity}. For sparsity, various papers have demonstrated different degrees of uniqueness based on the periodicity of the signal and other properties of the signal (e.g., collision-free in \cite{Ranieri_2013}). There is also interest in probabilistic analysis of the uniqueness of a random measurement matrix \cite{Eldar_2016_Hassibi_overview, Vinzant_2015_Injectivity, Candes_2015_Coded}.

The good news about non-trivial ambiguities is that although they play a big role in 1D phase retrieval, they have negligible effects in 2D. The following theorem from Hayes shows that the set of non-trivial ambiguities is measure zero.
\begin{theorem}[Hayes' Theorem \cite{Hayes_1982_Multidimensional, Hayes_1982_Reducible}: Non-trivial Ambiguities have measure zero in 2D]
Let $\vx \in \C^N$ be the ground truth, and define $\vy = |\mF \vx|^2 \in \R^{M}$ be the Fourier magnitude-square. If $d \ge 2$ and if the oversampling ratio is $s = M/N \ge 2$ (i.e., at least twice the Nyquist rate per dimension), then it is almost surely (a.s.) that there does not exist a non-trivial solution $\widehat{\vx}$ with $\vy = |\mF \widehat{\vx}|^2$ and $\widehat{\vx} \not\in \Omega(\vx)$.
\end{theorem}
\begin{proof}
    See \cite{Hayes_1982_Multidimensional, Hayes_1982_Reducible, bendory_2017_a}.
\end{proof}

\subsection{Fourier phase retrieval methods}
Assuming Hayes' result applies, much of the prior research then focuses on how to obtain the solution efficiently and robustly (against noise). Earlier approaches mostly follow the Gerchberg-Saxton algorithm \cite{Gerchberg_1972} where the solution is obtained by alternatingly projecting the estimate onto a subspace in the time domain and constraining the Fourier magnitude. Follow-up works of Fienup \cite{Fienup_1982_Comparison,Fienup_1986_Stagnation,Fienup_1987_Astronomy,Fienup_1990_Numerical,Fienup_1992_Joint} introduced additional spatial restrictions and projections. The Hybrid Input-Output (HIO) algorithm by Fienup is still a popular algorithm today, even though the algorithm is not guaranteed to converge, and when it converges, it is not necessarily a global minimum.

\begin{figure*}[t!]
    \centering
    \begin{tabular}{cc}
        \begin{tabular}{ccc}
            $\vphi$ & $\widehat{\vphi}$ & $\vphi - \widehat{\vphi}$ \\
            \includegraphics[width=0.12\linewidth]{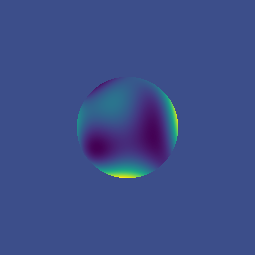} & \includegraphics[width=0.12\linewidth]{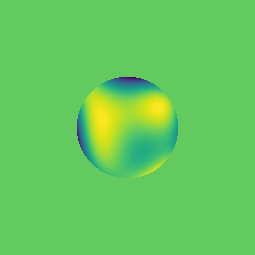} & \includegraphics[width=0.12\linewidth]{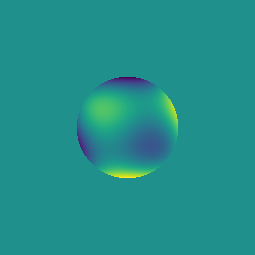} \\      
            $\abs{\mF \vx}^2$ & $\abs{\mF \widehat{\vx}}^2$ & $\abs{\mF \widetilde{\vx}}^2$ \\
            \includegraphics[width=0.12\linewidth]{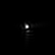} & \includegraphics[width=0.12\linewidth]{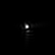} & \includegraphics[width=0.12\linewidth]{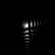}
        \end{tabular} &
        \begin{tabular}{ccc}
            $\vphi$ & $\widehat{\vphi}$ & $\vphi - \widehat{\vphi}$ \\ 

            \includegraphics[width=0.12\linewidth]{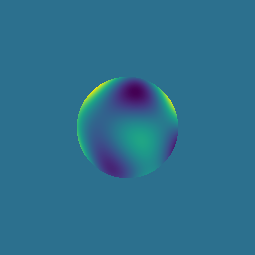} & \includegraphics[width=0.12\linewidth]{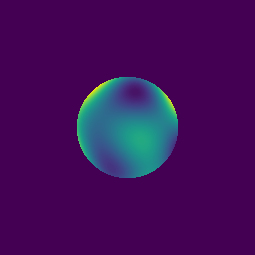} & \includegraphics[width=0.12\linewidth]{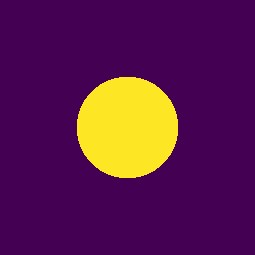} \\
            $\abs{\mF \vx}^2$ & $\abs{\mF \widehat{\vx}}^2$ & $\abs{\mF \widetilde{\vx}}^2$ \\
            \includegraphics[width=0.12\linewidth]{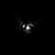} & \includegraphics[width=0.12\linewidth]{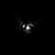} & \includegraphics[width=0.12\linewidth]{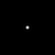}
        \end{tabular} \\
        (a) Correction by phase of $\widehat{\vx} = \calT_{\text{flip}}(\vx)$ & (b) Correction by phase of $\widehat{\vx} = \calT_{\text{dc}}(\vx)$
    \end{tabular}
    \caption{Comparing the equivalence classes of phase retrieval and wavefront estimation. Phase aberrations [Top] are shown with their resulting PSFs [Bottom] with $\vy = |\mF\vx|^2$ where $\vy$ is the PSF and $\vphi$ is the phase of $\vx$. In (a) we can see the PSFs match, though the optical correction does not achieve near-diffraction-limit performance. In (b) the PSFs also match, but $\vphi$ and $\vphi'$ only differ by an additive constant over the support. Thus, the optical correction achieves diffraction-limited performance.}
    \label{fig: pr_vs_we_equiv}
\end{figure*}

Beyond Gerchberg-Saxton and Fienup which belong to the category of projection algorithms \cite{Sayre_2002_Xray, Oszlanyi_2004_Abinitio, Palatinus_2013_Flipping}, there are other directions including gradient-based methods \cite{Yeh_2015_Ptychography, Candes_2015_Wirtinger, Fienup_2008_Transverse}, convex relations methods which lift the phase retrieval to a matrix completion problem \cite{Candes_2013_PhaseLift,Candes_2013_MatrixCompletion, Hassibi_2012_ISIT, Goldstein_2018_PhaseMax,Mallat_2015_MaxCut}, spectral methods to initialize the estimation \cite{Candes_2015_Wirtinger, YueLu_2019_SpectralInit, YueLu_2022_PMLR, Valzania_2021_SpectralInit}, and Bayesian methods such as approximate message passing \cite{Schniter_2016_AMP} and belief-propagation \cite{Montanari_Book}.

\subsection{Other known techniques}
\textbf{Coded Imaging}. A significant amount of the literature uses multiple \emph{coded} measurements pioneered by Cand\`{e}s \cite{Candes_2013_PhaseLift} and other researchers \cite{Hassibi_2012_ISIT, Goldstein_2018_PhaseMax,Mallat_2015_MaxCut}. Coded sampling can be achieved in various ways such as (i) masking \cite{Johnson_2008_Diffractive, Levin_2007_CodedAperture}, (ii) optical grating \cite{Loewen_Book, Asif_2016_FlatCam}, (iii) oblique illumination \cite{Faridian_2010_Nano}. During the sensing process, the coding scheme (e.g., the mask) is shuffled to ensure a certain degree of de-correlation across the \emph{multiple} measurements. This will then allow the reconstruction algorithm to achieve a globally unique solution. For us, since we are interested in real-time applications, we aim to perform wavefront estimation using a single measurement.

\textbf{Neural Wavefront Sensing}. A relatively new approach is neural wavefront shaping by Feng. et al. \cite{Feng_2023_NeuWS} where the wavefront signal is encoded by a neural radiance field. By jointly optimizing the reconstruction neural network and the phase mask, wavefront estimation can be achieved. Neural wavefront shaping relies on redundant measurements; it is reported in \cite{xie_2024_a} that they need at least 16 different spatial light modulator (SLM) patterns (often tens to hundreds).

\textbf{Deep Learning}. Deep learning approaches are plentiful. But if we closely inspect these methods, they are conceptually very similar to traditional methods, except that some steps are implemented through neural networks. For example, when solving the optimization, instead of resorting to classical methods such as ADMM, one can use Plug-and-Play \cite{Kamilov_2023_tutorial} by integrating pre-trained image denoisers to improve the prior. If we want more explicit control over the prior, generative priors can be used \cite{Hand_2018_GAN, Bohra_2022_Bayesian}. One can also replace the entire optimization-based approach with an end-to-end neural network, such as \cite{Metzler_2018_prDeep, Kappeler_2017_PtychNet}. A recent tutorial, Lam et al. \cite{EdmundLam_2024_Review}, has an excellent taxonomy of deep learning methods for phase retrieval to which we refer interested readers.

\section{Uniqueness of Wavefront Estimation}

\subsection{What is wavefront estimation?}
\label{sec: what_is_we}
To motivate the wavefront estimation problem, we need to briefly describe how adaptive optics work. When an image is formed by an optical system, the response of the system is described by its point spread function (PSF). The PSF is said to be at the diffraction limit when it is free of aberrations. Aberrations are introduced when imaging through random media or in the presence of sub-optimal optical components.

The goal of adaptive optics is to compensate for these aberrations in real-time. Let $\phi[\vn]$ be the ground truth and $\widehat{\phi}[\vn]$ be the estimate. The mathematical operation of this compensation is then as follows:
\begin{equation}
    \underset{\text{compensated signal}}{\underbrace{\widetilde{x}[\vn]}} \bydef \underset{\text{original $x[\vn]$}}{\underbrace{\abs{x[\vn]} e^{j \phi[\vn]}}} \times \underset{\text{compensation}}{\underbrace{e^{-j \widehat{\phi}[\vn]}}}.
\end{equation}
The quality of this operation is then determined by the resulting PSF, i.e., $\abs{\mF \vxtilde}^2$. If $\abs{\mF \vxtilde}^2$ is a diffraction-limited spot, then the compensation can be declared a success.

\textbf{Example 1.} [Conjugate Flip] Consider a ground truth phase $\phi[\vn]$ and the corresponding signal $x[\vn] = |x[\vn]|e^{j\phi[\vn]}$. We apply a conjugate flip trivial ambiguity to $x[\vn]$ by considering $\widehat{\phi}[\vn]  = -\phi[-\vn]$ so that $\widehat{x}[\vn] = x[\vn] e^{j \widehat{\phi}[\vn]} = \calT_{\text{flip}}(x[\vn]) = x^*[-\vn]$. It is easy to show that $|\mF\vx| = |\mF\widehat{\vx}|$. An adaptive optics system that uses this $\widehat{\vphi}$ to compensate for the phase $\vphi$ producing $\vphi - \widehat{\vphi}$ is shown in \fref{fig: pr_vs_we_equiv}(a). If we take the corresponding Fourier magnitude square of the compensated signal, we will see a worsened PSF.

\textbf{Example 2.} [DC Offset] Consider a ground truth phase $\phi[\vn]$ and the corresponding signal $x[\vn] = |x[\vn]|e^{j\phi[\vn]}$. This time we consider a DC offset trivial ambiguity $\widehat{\phi}[\vn] = \phi[\vn] + \phi_0$ so that $\widehat{x}[\vn] = x[\vn] e^{j \widehat{\phi}[\vn]} = \calT_{\text{dc}}(x[\vn]) = x[\vn]e^{j\phi_0}$. We can again show that $|\mF\vx| = |\mF\widehat{\vx}|$. Using an adaptive optics correction with this $\widehat{\vphi}$ results in a phase profile as shown in \fref{fig: pr_vs_we_equiv}(b). If we take the corresponding Fourier magnitude square of the compensated signal, we will see a diffraction-limited spot.

In summary, the conjugate flip degrades our system while a DC offset is immaterial. This leads to the question: How do we check whether a candidate solution $\widehat{\vx}$ is worsening the imaging system? One way to do this is to consider a criterion:
\begin{equation}
\underset{|\vx^H\vx|, \text{ ideal PSF}}{\underbrace{\sum_{\vn} |x[\vn]|^2 }} \overset{?}{=} \underset{\vx^H \widehat{\vx}, \text{ compensated PSF}}{\underbrace{\sum_{\vn} |x[\vn]|^2e^{j(\phi[\vn]-\widehat{\phi}[\vn])} }}
\label{eq: intuition}
\end{equation}
This equation checks whether the complex conjugate $\vx^H$ is ``aligned'' with the estimated phase $\widehat{\vx}$. If it does, then the product will give us the magnitude square of the signal. It is easy to show that the conjugate flip does not satisfy this criterion whereas the DC offset does.

Based on the intuition given in \eqref{eq: intuition}, we define the wavefront estimation problem as follows:
\begin{definition}[Wavefront Estimation]
Let $\vx \in \C^N$ be the ground truth complex signal, and define $\vy = \abs{\mF \vx}^2 \in \R^M$ be the Fourier magnitude. We define the wavefront estimation problem as
\begin{equation}
    \text{Find} \; \widehat{\vx} \; \text{such that} \quad \vy    = \abs{\mF \widehat{\vx}}^2 \; \text{and} \; \abs{\vx^H \widehat{\vx}} = |\vx^H\vx|.
    \label{eq: we_problem}
\end{equation}
We call $\abs{\vx^H \widehat{\vx}} = |\vx^H\vx|$ the \textbf{optical correction constraint}.
\end{definition}
Comparing \eqref{eq: we_problem} and \eqref{eq: pr_matrix}, we recognize that while both phase retrieval and wavefront estimation aim to recover the phase from Fourier magnitude, wavefront estimation requires an optical correction constraint to reject trivial ambiguities. Therefore, wavefront estimation is a sub-problem of phase retrieval.

\subsection{Geometry of pupil}
The shape of the pupil plays a big role in wavefront estimation. We define the pupil function below.
\begin{definition}[Pupil]
    The pupil $\vp \in \R^N$ is an indicator function defined such that
    \begin{equation}
        p[\vn] = \begin{cases}
            1 & \abs{x[\vn]} \neq 0 \\
            0 & \abs{x[\vn]} = 0
        \end{cases}, \; \forall \vn \in \Z_N.
    \end{equation}
\end{definition}
Therefore, the vector $\vp$ denotes a binary mask. In the absence of any phase aberration, i.e., $\phi[\vn] \equiv 0$ for all $\vn$, the shape of the pupil will determine the diffraction-limited spot.

\begin{definition}[Diffraction limited PSF]
    Let $\vx \in \C^N$ and $x[\vn] = p[\vn] e^{j \phi[\vn]}$ be the system response, i.e., the magnitude is $|\vx| = \vp$. If $\phi[\vn] \equiv 0$ for all $\vn$, then the PSF given by
    \begin{equation}
        \vy = \abs{\mF \vp}^2
    \end{equation}
    is called a diffraction-limited PSF.
\end{definition}

The diffraction-limited PSF provides us with a convenient way to measure the quality of the estimated wavefront. In the optics literature, the most commonly used metric is known as the Strehl ratio \cite{Tyson_Frazier_Book}.
\begin{definition}[Strehl Ratio]
Let $\widehat{\vx}$ be the estimated optical aberration. We define the Strehl ratio as
\begin{equation}
        q \bydef \max\left\{ \abs{\mF \widehat{\vx}}^2 \right\} / \max \left\{ \abs{\mF \vp}^2 \right\}.
        \label{eq: strehl}
\end{equation}
where $q$ is a number between $0 \le q \le 1$.
\end{definition}
A Strehl ratio of $q = 1$ corresponds to a diffraction limited PSF while a Strehl ratio $q \ll 1$ is a severely distorted PSF.

\subsection{Main result 1: Impact of the constraint}
We now present our first main result regarding the wavefront estimation problem formulation \eqref{eq: we_problem}. We show that while there are many trivial ambiguities of the original phase retrieval problem in \eqref{eq: pr_matrix}, by adding the optical correction constraint in \eqref{eq: we_problem} the set of solutions reduces to complex scalar multiples of the true solution.

\begin{theorem}
\label{th: theorem2}
Let $\vx \in \C^N$ be the ground truth complex signal, and let $\widehat{\vx} \in \Omega(\vx)$ be one of the trivial solutions obtained by solving the original phase retrieval \eqref{eq: pr_matrix}. If $\widehat{\vx}$ satisfies the optical correction constraint $|\vx^H \widehat{\vx}| = |\vx^H\vx|$ defined in \eqref{eq: we_problem}, then it is necessary that
\begin{equation}
\widehat{x}[\vn] = x[\vn] e^{j c},
\end{equation}
for some $c \in [0, 2\pi)$.
\end{theorem}
\begin{proof}
See Appendix.
\end{proof}
The following two Corollaries discuss the impact of using solutions to \eqref{eq: we_problem} in the context of phase conjugate AO.

\begin{corollary}
\label{corollary 1}
Let $\vx \in \C^N$ be the ground truth complex signal, and let $\widehat{\vx} \in \Omega(\vx)$ be the solution of \eref{eq: we_problem}. Then, if we compensate $\vx$ by multiplying the phase component of $\widehat{\vx}$, the compensated signal is
\begin{equation*}
\widetilde{x}[\vn] = |x[\vn]|e^{j\phi[\vn]} \times e^{-j\widehat{\phi}[\vn]} = |x[\vn]| e^{-jc},
\end{equation*}
for some $c \in [0,2\pi)$.
\end{corollary}
\begin{proof}
Theorem~\ref{th: theorem2} implies that $\widehat{x}[\vn] = x[\vn] e^{j c}$, which gives us $\widehat{\phi}[\vn] = \phi[\vn] + c$ for $c \in [0, 2\pi)$. Thus $\forall \vn \in \Z_N$,
\begin{align*}
\widetilde{x}[\vn]
&\bydef |x[\vn]|e^{j\phi[\vn]} \times e^{-j\widehat{\phi}[\vn]} \\
&= |x[\vn]|e^{j(\phi[\vn] - \phi[\vn] - c)} = \abs{x[\vn]} e^{-jc}.
\end{align*}
\end{proof}

In the special case where the magnitude of $\vx$ is just the pupil function, i.e., $\vx = \vp e^{j\vphi}$, the solution of \eref{eq: we_problem} will give us the ideal Strehl ratio of $q = 1$ as summarized in Corollary 2.
\begin{corollary}
Let $\vx \in \C^N$ be the ground truth complex signal, and suppose that $x[\vn] = p[\vn] e^{j \phi[\vn]}$. Let $\widehat{\vx} \in \Omega(\vx)$ be the solution of \eref{eq: we_problem}. Then, the phase compensated signal $\widetilde{\vx}$ will have a Strehl ratio of $q = 1$.
\end{corollary}
\begin{proof}
Corollary~\ref{corollary 1} implies that $\vxtilde = \abs{\vx} e^{-jc}$. Since $|\vx| = \vp$, we have $\vxtilde = \vp e^{-jc} $ for $c \in [0, 2\pi)$. Thus, by definition of the Strehl ratio, the phase compensated signal $\widetilde{\vx}$ will have a Strehl ratio of
\begin{align*}
         q &= \max\left\{ \abs{\mF \vxtilde}^2 \right\} / \max \left\{ \abs{\mF \vp}^2 \right\} \\
         &= \max\left\{ \abs{e^{-jc} \mF \vp}^2 \right\} / \max \left\{ \abs{\mF \vp}^2 \right\} = 1,
\end{align*}
since $e^{-jc}$ can be eliminated.
\end{proof}

These results indicate that if we can find a solution to \eqref{eq: we_problem}, we will recover the solution $\vx$ up to a DC shift, i.e., a multiplication with a complex exponential. Furthermore, if we can recover $\vx$ up to this set, the resulting phase conjugation can be said to be optimal.

\subsection{Main result 2: invertible PSF}
\label{sec: big_thm}

Hayes' Theorem and Theorem~\ref{th: theorem2} have little to do with the ease of recovery because one cannot resolve the trivial ambiguities. We now turn towards defining an invertible PSF and presenting the second key Theorem of this paper. We derive a set of conditions under which a PSF is invertible. The implication of these constraints is a set of \emph{imaging systems} that produce \emph{only} PSFs whose phase can be uniquely estimated.

\begin{definition}[Invertible PSF]
Let $\vx \in \C^N$ and $x[\vn] = p[\vn] e^{j \phi[\vn]}$ with $\vn \in \Z_N$. If for every $\widehat{\vx} \neq \vx$ that satisfies $\abs{\mF \widehat{\vx}} = \abs{\mF \vx}$ the optical correction constraint $\abs{\vx^H \widehat{\vx}} = |\vx^H\vx|$ is also satisfied, then we say $\vy = \abs{\mF \vx}^2$ is an invertible PSF.
\end{definition}
The following Theorem describes how to design a system that produces only invertible PSFs.

\begin{theorem}[Existence of invertible PSF imaging system]
\label{th: theorem_iPSF}
Let $\abs{x[\vn]} > 0$ for some subset of $\Z_N$. Suppose that the measurement $y[\vm]$ with $d \geq 2$ is oversampled by a factor $s \geq 2$ under measurement model \eqref{eq: measurement_model}. The set of trivial ambiguities $\Omega(\vx)$ is reduced to the unique solution $\{\vx\}$ (a.s.) if all of the following conditions are met:
\begin{enumerate}[(i)]
\item The pupil $p[\vn]$ is known a priori $\forall \vn \in \Z$;
\item The shifted pupil is aperiodic: $\calT_{\text{shift}}(\vp) \not= \vp$ for any amount of non-zero (circular) shift;
\item The flipped pupil is aperiodic: $\calT_{\text{shift}} \circ \calT_{\text{flip}}(\vp) \not= \vp$ for any amount of (circular) shift;
\item The average value of the phase offset by $\calT_{\text{dc}}$ is known.
\end{enumerate}
\end{theorem}

\begin{proof}
See Appendix.
\end{proof}

Under the above conditions, the corresponding phase retrieval has \emph{only} one solution, i.e., $\widehat{\vx} = \vx$, almost surely (a.s.). Furthermore,     with the exception of $\vx \neq \vzero$ and condition (iv), the constraints are all on the \textit{geometry} of the pupil $\vp$ which, in the context of optics, can be \textit{designed}.

\textbf{Understanding the Geometry}. Among the four conditions above, the most interesting ones are (ii) and (iii). We use \fref{fig: shape} to illustrate the idea. The first support shown satisfies all conditions, the second does not satisfy condition (ii), and the final two do not satisfy condition (iii).

\begin{figure}[h]
\centering
\includegraphics[width=0.95\linewidth]{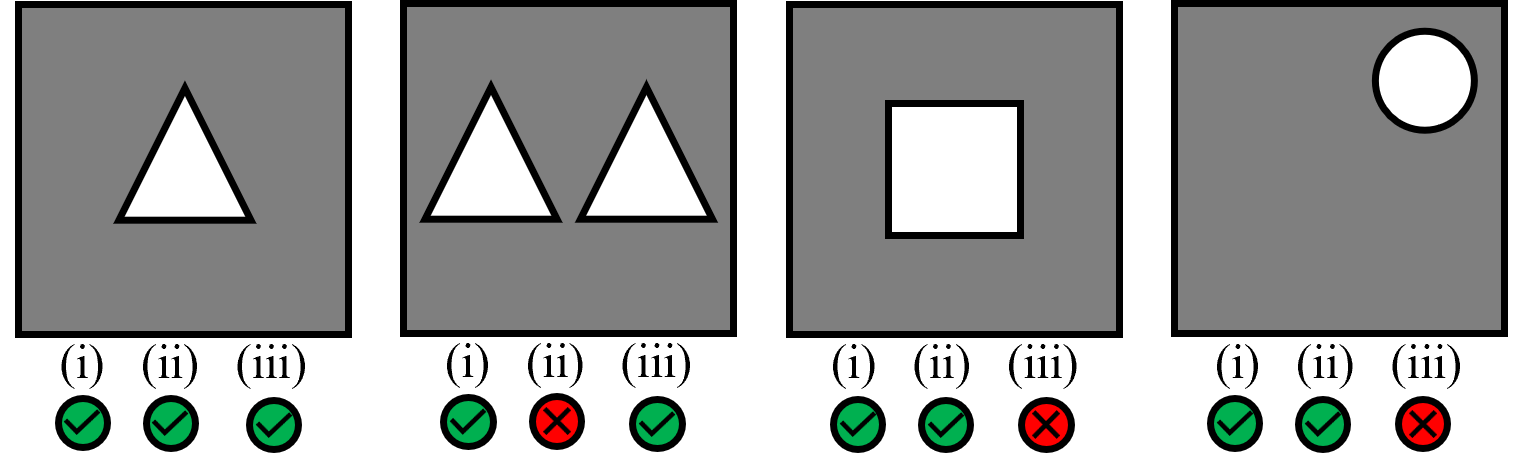}
    \vspace{-2ex}
\caption{Visualizing various pupil geometries that do and do not satisfy pupil geometry conditions (i-iii) of Theorem~\ref{th: theorem_iPSF}.}
\label{fig: shape}
\end{figure}

\textbf{Toy Example}. We use an example to highlight the impact of our conditions on the inverse problem \eqref{eq: we_problem}. Consider a phase that is the sum of two astigmatism Zernike polynomials, $\phi[\vn] = a_5 Z_5[\vn] + a_6 Z_6[\vn]$, with ground truth $a_5 = a_6 = 2$ and $Z_\ell[\vn]$ is the ${\vn}$th sample of the $\ell$th Zernike polynomials \cite{Noll_1976, Chan_TurbulenceBook}. We construct the signal as $x[\vn] = p[\vn] e^{j \phi[\vn]}$. Given the Fourier magnitude $\vy = |\mF\vx|^2$, we want to recover $(a_5, a_6)$.

We consider four different pupil shapes as shown in \fref{fig: non_convexity_toy}. The shapes of the pupils are known a priori, and they will replace $p[\vn]$ when defining $x[\vn]$. We are interested in (i) finding the phase retrieval solution by solving the original problem \eqref{eq: pr_matrix} using different shapes; and (ii) inspecting the optimization landscape, i.e., the objective function.

\begin{figure}[h!]
    \centering
    \includegraphics[width=0.95\linewidth]{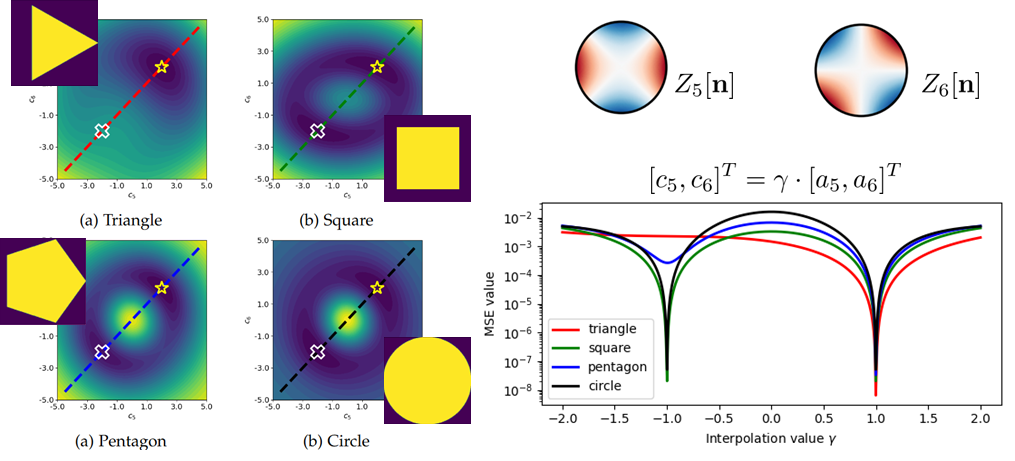}
    \vspace{-2ex}
    \caption{Toy example illustrating the impact of pupil geometry on recovery.}
    \label{fig: non_convexity_toy}
\end{figure}

\fref{fig: non_convexity_toy}(a)-(d) shows the objective function which is a 2D function of variables $(a_5,a_6)$. The true solution is $(a_5, a_6) = (2,2)$ (marked as a star), and the conjugate flip is $(a_5,a_6) = (-2,-2)$. We cut the cross section along the diagonal, and we plot the energy level on the right-hand side of \fref{fig: non_convexity_toy}. To make our plots consistent across the four cases, we use a dummy parameter $\gamma$ to control the sweep between solutions and plot along $\gamma (a_5, a_6)$, hence $\gamma = 1$ is the true solution case and $\gamma = -1$ is the conjugate flip case.

Based on the plot, we can make the following observations:
\begin{itemize}
\item \textbf{Circle}: There are two equally good global minima, no matter what algorithm we use there is a 50\% chance we will hit the correct solution and a 50\% chance the conjugate flip. We cannot distinguish the two.
\item \textbf{Square}: The behavior is very similar to a circle where there are two equally good global minima. Again, we cannot distinguish the two.
\item \textbf{Triangle}: There is only one local minimum, and this local minimum \emph{is} the correct solution.
\item \textbf{Pentagon}: While we can see two local minima, their objective values are different. This implies that as long as we have a reasonably designed algorithm that can search for a better local minimum, we will get a unique solution.
\end{itemize}
In summary, if we use a non-centrosymmetric pupil (such as a triangle or a pentagon), we can identify the correct solution.

\subsection{Hasn't symmetry breaking been used for a long time?} In optics, symmetry breaking has been reported by Fienup \cite{Fienup_1983_Reference, Fienup_1986_Stagnation, Fienup_1983_Uniqueness, Fienup_1982_Support, Fienup_1989_Wavefront} in the 1980's. The difference is that in Fienup's work, the approach is largely about finding the appropriate boundary conditions of the autocorrelation function \cite{Fienup_1986_Boundary}. Because the approach boils down to solving a system of linear equations, it is unclear how one can readily generalize the concept to analyze the optimal geometry or design neural networks. Other works such as \cite{Martinache_2013_Asymmetric} and \cite{Paxman_2019_Ambiguity} have specific choices of the aperture but they are limited to highly customized problems.

In signal processing, interestingly, there are limited discussions about symmetry breaking. The most relevant comments were made by Sun and collaborators \cite{JuSun_2020_BreakSymmetry, JuSun_2021_BreakSymmetry} where they observed the same phenomenon we present here, but their approach introduces constraints in the training dataset. For us, instead of manipulating the data we re-design the optics.

\subsection{Going from uniqueness in $\vx$ to $\vphi$}
Up to this point, we have considered the recovery of $\vx \in \C^N$, though we are often more interested in finding $\vphi \in \R^N$. However, the uniqueness in $\vx$ does not guarantee uniqueness in $\vphi$. To achieve uniqueness in $\vphi$ by uniqueness in $\vx$, we must meet two (additional) requirements related to phase unwrapping:
\begin{enumerate}
    \item The difference between neighboring samples of $\phi[\vn]$ can not exceed $\pm\pi$, thus the phase can be recovered up to an average value by phase unwrapping. One can interpret this as a smoothness constraint.
    \item The average value of $\phi[\vn]$ is 0, i.e., $\vp^T \vphi = 0$. Note we have already assumed this WLOG previously.
\end{enumerate}
Under these two conditions, the uniqueness in $\vx$ (from Theorem~\ref{th: theorem_iPSF}) can be transferred to $\vphi$ which follows from the phase unwrapping problem being uniquely solvable \cite{itoh_1982_a}.

\section{Wavefront Estimation Algorithms}
\label{sec: method}
Now that we understand how to design our pupil, the next big question is how to rapidly estimate the phase. In this section, we propose a few wavefront estimation algorithms for non-circular apertures. We additionally describe the use of neural representations to calibrate for the shape.

\subsection{Proposed methods for phase recovery}
In developing a wavefront estimation algorithm, uniqueness provides a significant benefit. If the problem has a unique solution, a network can be trained in a supervised fashion to invert PSFs. On the other hand, if an imaging system does not have an invertible PSF, then the network will not generalize as shown in \fref{fig: s2p_doesnt_work 1}. 

In what follows, we present a few \emph{efficient} algorithms that can invert PSFs in real time. Our algorithms are inspired by previous methods in the literature. We do not necessarily advocate for one method over another. Our goal here is to demonstrate that if the solution is unique then one has a lot of flexibility in designing the network. We consider three classes of algorithms: (i) direct methods, (ii) unrolled methods and (iii) neural representations.

\subsubsection{Proposed method 1 (inspired by \cite{mao2021accelerating}): Direct PSF inversion}
A PSF is band-limited by definition. Because of this, we consider writing a PSF as a linear combination of basis functions:
\begin{equation}
y[\vn] = \sum_{k=1}^K \beta_{k} \zeta_{k}[\vn],
\end{equation}
where $\beta_k \in \R$ is the $k$th basis coefficient and $\zeta_{k}[\vn]$ is the corresponding basis function, which can be obtained offline using supervised learning \cite{mao2021accelerating}. For notational simplicity we write $\vbeta = \{\beta_1,\ldots,\beta_K\}$. Similarly, we write $\phi[\vn]$ as
\begin{equation}
\phi[\vn] = \sum_{\ell=1}^L \alpha_{\ell} \xi_{\ell}[\vn],
\end{equation}
where $\xi_\ell[\vn]$ is an orthogonal basis over $\vp$, and $\valpha = \{\alpha_1,\ldots,\alpha_L\}$. We choose $\vxi$ to be the Zernike polynomials when $\vp$ is a circle.

While $\valpha \to \vbeta$ was previously considered in \cite{mao2021accelerating}, known as the Phase-to-Space (P2S) transform for simulating imaging through turbulence for simulation \cite{Chimitt_2020_OpEng, Chimitt_2022_RealTime}, wavefront estimation requires $\vbeta \to \valpha$. As discussed previously, $\vbeta \to \valpha$ is unique if Theorem~\ref{th: theorem_iPSF} is satisfied. We refer to $\vbeta \to \valpha$ as the Space-to-Phase (S2P) transform, which we approximate with a learned representation, $\stop_\theta(\vbeta) \approx \valpha$, using a lightweight 3-layer perceptron (MLP) with ReLU activation functions. The S2P method is illustrated in \fref{fig: s2p}.

To successfully train the S2P, aside from satisfying Theorem~\ref{th: theorem_iPSF}, we also find using basis functions that are orthogonal over the designed pupil $\vp$ helps significantly. For flexibility, we mask the Zernike polynomials by $\vp$ then perform Gram-Schmidt orthogonalization. \fref{fig: basis} shows various $\vxi$ basis functions obtained as described.

\begin{figure}[h]
    \centering
    \includegraphics[width=0.8\linewidth]{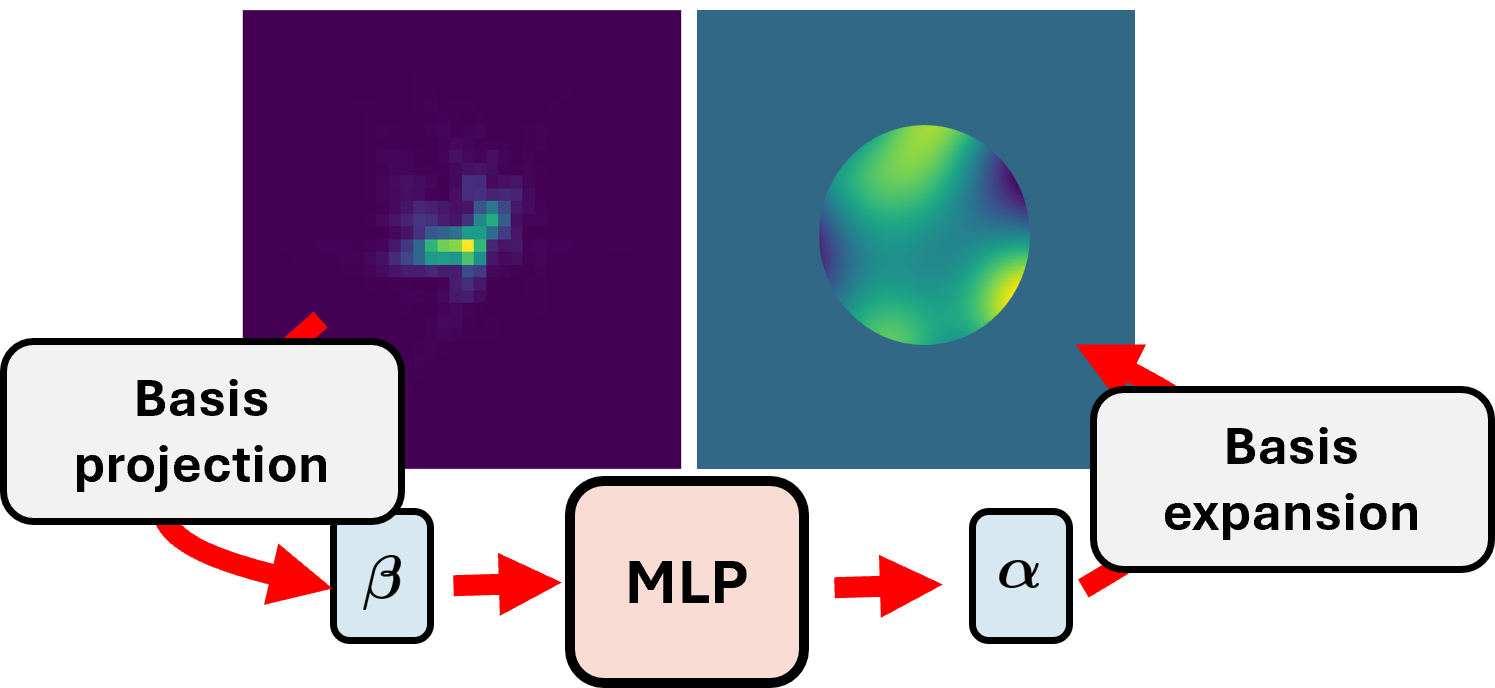}
    \vspace{-2ex}
    \caption{The Space-to-Phase (S2P) transform overview.}
    \label{fig: s2p}    
\end{figure}

\begin{figure}[h]
    \centering
    \includegraphics[width=0.9\linewidth]{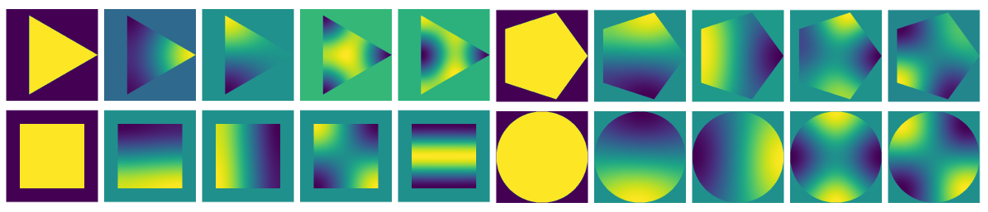}
    \vspace{-2ex}
    \caption{Visualization of phase basis functions for different pupil geometries.}
    \label{fig: basis}
\end{figure}   

The S2P has a few advantages:
\begin{enumerate}[(i)]
\item The S2P runs at $\sim\!5,200$ frames per second for a $128 \times 128$ PSF on an AMD Ryzen 7 5700X3D CPU. This speed can be attributed to the usage of the basis representations, reducing the problem dimensionality.
\item The S2P is robust to noise as high-frequency noise can not be captured in the low dimensional subspace spanned by the basis representation.
\end{enumerate}

Aside from representing \emph{both} the PSF and phase via basis representation, we additionally consider the mappings $\vy \to \valpha$ and $\vy \to \vphi$. Surprisingly, we find the MLP structure described for the S2P works for these mappings as well, speaking to the well-posed nature of recovery given uniqueness.

\subsubsection{Proposed method 2 (inspired by \cite{Gerchberg_1972, Fienup_1982_Comparison}): Unrolled PSF inversion}
Direct PSF inversion lacks a ``feedback mechanism'' such as those in the Gerchberg-Saxton \cite{Gerchberg_1972} and HIO algorithm by Fienup \cite{Fienup_1982_Comparison}. An interesting quality of the HIO algorithm is the intuition that the intermediate ``outputs'' of the algorithm need not be the estimation of the measurement, but a driving function for the iterative algorithm. Following the HIO algorithm, we consider an unrolled verison of the S2P in which the intermediate variables are not constrained to be wavefront estimations but simply driving functions for the final output. We summarize the unrolled S2P algorithm in Algorithm \autoref{alg: unrolled}.

\renewcommand{\algorithmiccomment}[1]{//#1}
\begin{algorithm}
\begin{algorithmic}
\REQUIRE Input: $\vy$, $\vp$, Output: $\widehat{\vphi}$
\STATE $\vy_0 \bydef \vy$
\FOR{$k = 1$ to $K$} 
    \IF{$k = 1$}
    \STATE $\valpha_{k} \leftarrow g_{\theta_k}\left( \vy_0 \right)$
    \ELSE
    \STATE $\valpha_{k} \leftarrow g'_{\theta_k}\left( [\vy_0, \vy_{k-1}, \valpha_{k-1}] \right)$ \hfill\COMMENT{\textit{Concatenated inputs}}
    \ENDIF
    \STATE $\phi_k[\vn] = \sum_{\ell} \alpha_{k,\ell} \xi_{\ell}[\vn]$
    \STATE $\vy_k \leftarrow \abs{\mF (\vp \odot e^{j \vphi_{k}})}^2$ \hfill\COMMENT{\textit{Can be accelerated via P2S}}
\ENDFOR
\RETURN $\widehat{\vphi} = \vphi_{K}$
\end{algorithmic}
\caption{Unrolled S2P network}
\label{alg: unrolled}
\end{algorithm}

\subsubsection{Proposed method 3 (inspired by \cite{Fienup_1993_Complicated, Feng_2023_NeuWS}): Neural representation}
Using a basis representation in the phase retrieval problem was considered by Fienup \cite{Fienup_1993_Complicated}. Using this representation, one can perform gradient-based optimization on the following loss:
\begin{equation}
    \widehat{\valpha} = \argmin{\valpha} \sum_{t=1}^T \norm{\vy_t - \abs{\mF \left( \vp \odot e^{j \sum_{\ell=1}^L \alpha_{\ell} \vxi_{\ell}} \right) }^2}^2,
    \label{eq: l2_we}
\end{equation}
where $\vxi = [\xi[\vn_1], \ldots, \xi[\vn_N]]^T$ was chosen to be the Zernike polynomials in \cite{Fienup_1993_Complicated}. Recently, the Zernike polynomials were used as an encoding for implicit neural representations (INR) was proposed in \cite{Feng_2023_NeuWS} used to iteratively solve a modified version of \eqref{eq: l2_we}. To fit our problem, we modify the Zernike encoding to be the $\vxi$-encodings shown in \fref{fig: basis}. We refer to this as a $\vxi$-Encoded INR.

For Fienup \cite{Fienup_1993_Complicated}, the Zernike representation reduced the dimensionality of the problem, offering some benefit over the standard GS and HIO algorithms, though the trade-off is that it relied on gradient-based methods making initialization crucial. The $\vxi$-encoded INR is a basis that is then passed through a non-linear transformation, hence it can be viewed as a generalization of a linear basis representation. We find that using the $\vxi$-encoded INR tends to perform well even in the presence of noise, likely due to the natural smoothness imposed by the INR. In the noiseless case, we will see the that the $\vxi$-Encoded INR performs at or near the performance of the GS and HIO algorithms.

\subsection{Neural representations for pupil shape calibration}
We now discuss the pupil calibration problem in which we no longer assume an exact knowledge of $p[\vn]$. Suppose we can illuminate the pupil with a known signal with $x[\vn] = c[\vn] e^{j \psi[\vn]}$ where $c[\vn]$ is the support of the calibration pattern and $\psi[\vn]$ is the phase of the calibration pattern over $c[\vn]$ with the support of $c[\vn]$ beyond that of $p[\vn]$. We define the following as the pupil identification problem:
\begin{definition}[Pupil Identification]
    Given $\boldsymbol{\psi}$ and $\vy$,
    \begin{equation}
        \text{Find} \;\; \vp \quad \text{subject to} \quad \vy = \abs{\mF (\vp \odot e^{j \vpsi})}^2,
        \label{eq: pupil_cal_problem}
    \end{equation}
    where $\odot$ represents the elementwise product of two vectors.
\end{definition}

\begin{theorem}[Uniqueness of Pupil Shape Calibration]
\label{th: theorem_calibration1}
    Let calibration pattern $\psi[\vn]$ be known over support $c[\vn]$ and measurement $y[\vn]$ satisfy Hayes' Theorem. If there does not exist an $\vell \in \Z_N$ such that $p[\pm \vn + \vell] \neq p[\vn]$ and $p[\pm \vn + \vell] e^{j\psi[\vn]} \in \Omega(\vx)$ then the pupil shape $p[\vn]$ that satisfies \eref{eq: pupil_cal_problem} is unique (a.s.).
\end{theorem}
\begin{proof}
See Appendix.
\end{proof}

This theorem can be stated intuitively: if there is no other way to position the pupil to produce the same PSF, $p[\vn]$ is unique. This imposes two requirements on $\psi[\vn]$ if we want to identify the pupil shape using a single measurement. The first requirement is that for a single calibration pattern, $\psi[\vn]$ can not be too ``simple''. Second, it imposes a joint requirement on the shape of $p[\vn]$ and the chosen pattern $\psi[\vn]$.

\textbf{Neural mapping implementation}. For optimizing the geometry of the pupil, we use a coordinate-encoded MLP with a standard positional Fourier encoding \cite{mildenhall_2021_a} to represent the pupil. Using a coordinate-encoded MLP is useful since the pupil is typically not a collection of delta functions, but some shape with a specific structure. Denoting
\begin{equation}
     x^{(m)}_{\theta}[\vn] = p_{\theta}[\vn] e^{j \psi^{(m)}[\vn]},
\end{equation}
where $m$ indicates the corresponding measurement, the MLP aims to solve \eqref{eq: pupil_cal_problem} via $\ell^2$ loss.
The final $\vp_{\widehat{\theta}}$ is thresholded to give an estimate of the pupil. We display a qualitative result for pupil recovery in \fref{fig: purdue} where our goal is to recover $\vp$ from $M=\{1, 2, 10\}$ measurements.

\begin{figure}
    \centering
    \begin{tabular}{cccc}
        \includegraphics[width=0.2\linewidth]{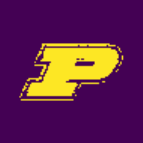} & \hspace{-2ex} \includegraphics[width=0.2\linewidth]{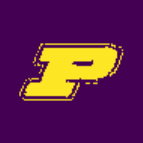} & \hspace{-2ex}\includegraphics[width=0.2\linewidth]{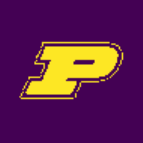} & \hspace{-2ex}\includegraphics[width=0.2\linewidth]{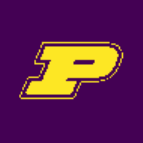} \\
        $\vp_{\widehat{\theta}}$ ($M = 1$) & \hspace{-2ex} $\vp_{\widehat{\theta}}$ ($M = 2$) &\hspace{-2ex}
        $\vp_{\widehat{\theta}}$ ($M = 10$) &\hspace{-2ex} $\vp$ (ground truth)
    \end{tabular}
    \caption{Qualitative performance for pupil calibration problem. The subfigures above show the estimated pupil $\vp_{\widehat{\theta}}$. Here, $M$ is the number of measurements.}
    \label{fig: purdue}
\end{figure}

\section{Experiments}
To help further distinguish between the problem of phase retrieval and wavefront estimation, we will compare the performance of various algorithms on each problem. As we will see, high performance in phase retrieval does not always transfer to wavefront estimation. We additionally built a prototype triangle pupil system and discuss some results.

\subsection{Simulation setting}
Our simulated experiments consist of drawing i.i.d. uniform random variables for Zernike coefficients in $\valpha$, $\alpha_{i} \sim U(-2, 2)$, with $\vphi$ formed by a weighted sum of the Zernike polynomials and windowed by the corresponding $\vp$. The PSF $\vy$ is then generated and noise is optionally added and clipped as $\text{max} (\vy + \vn, 0)$ where $\vn \sim \calN(0, \sigma^2 \mI)$. We choose $\vp$ for all shapes to (approximately) have the same amount of support pixels to keep resolution consistent.

Regarding noise level, a PSF is highly concentrated in the center of the grid, making a typical signal-to-noise ratio (SNR) calculation potentially misleading. The noisy case reported in Table~\ref{tab:big_exp} has an SNR of $-1$dB when computed over the center $60\times60$ crop and $-7.5$dB over the entire $128\times128$ grid. The S2P and related methods can operate directly on this center crop, allowing them some flexibility to ignore noise beyond the typical support of a PSF. For fairness, we remove the noise beyond the $60\times60$ crop when using the GS/HIO algorithms.

\begin{table*}
    \centering
    \caption{Performance on phase retrieval and wavefront estimation problems for various support shapes with and without noise, with speed comparisons. The best and second-best performing algorithms are indicated with cell shading and underlined, respectively. Algorithms proposed in this paper are indicated with a $\dagger$.}
    \label{tab:big_exp}
    
    \begin{tabular}{l|l|ccc|ccc|c}
        \toprule
        \multicolumn{9}{c}{\textbf{Task: \textit{Phase retrieval (trivial ambiguities do not matter).} Metric: \textit{Normalized correlation.}}} \\
        \midrule
        & & \multicolumn{3}{c|}{\textbf{Iterative methods}} & \multicolumn{3}{c|}{\textbf{Direct methods}} & \textbf{Other} \\
        \textbf{Support} & \textbf{Condition} & \textbf{GS} \cite{Gerchberg_1972} & \textbf{HIO} \cite{Fienup_1982_Comparison} & \textbf{$\vxi$-Encoded}$^\dagger$ & \textbf{S2P}$^\dagger$ & $(\vy \to \valpha)^\dagger$ & $(\vy \to \vphi)^\dagger$ & \textbf{Unrolled S2P}$^\dagger$ \\
        \midrule
        \multirow{2}{*}{Circle} & Noiseless & \cellcolor[gray]{0.9}$0.996 \pm 0.011$ & $\underline{0.996 \pm 0.014}$ & $0.993 \pm 0.033$ & $0.906 \pm 0.096$ & $0.913 \pm 0.095$ & $0.907 \pm 0.110$ & $0.928 \pm 0.087$ \\
        & Noisy & $\underline{0.919 \pm 0.033}$ & $0.901 \pm 0.055$ & \cellcolor[gray]{0.9}$0.986 \pm 0.015$ & $0.876 \pm 0.122$ & $0.887 \pm 0.114$ & $0.899 \pm 0.109$ & $0.869 \pm 0.132$ \\
        \midrule
        \multirow{2}{*}{Triangle} & Noiseless & $0.992 \pm 0.017$ & \cellcolor[gray]{0.9}$1.000 \pm 0.000$ & $0.989 \pm 0.031$ & $0.979 \pm 0.036$ & $0.987 \pm 0.030$ & $0.989 \pm 0.028$ & $\underline{0.996 \pm 0.004}$ \\
        & Noisy & $0.900 \pm 0.037$ & $0.881 \pm 0.062$ & \cellcolor[gray]{0.9}$0.981 \pm 0.22$ & $0.972 \pm 0.027$ & $0.971 \pm 0.034$ & \underline{$0.973 \pm 0.028$} & $0.965 \pm 0.040$ \\
        \midrule
        \multirow{2}{*}{Pentagon} & Noiseless & $0.991 \pm 0.012$ & \cellcolor[gray]{0.9}$1.000 \pm 0.005$ & $0.993 \pm 0.017$ & $0.979 \pm 0.034$ & $0.996 \pm 0.009$ & $0.996 \pm 0.010$ & $\underline{0.997 \pm 0.003}$ \\
        & Noisy & $0.912 \pm 0.033$ & $0.893 \pm 0.057$ & \cellcolor[gray]{0.9}$0.983 \pm 0.015$ & $0.940 \pm 0.081$ & $0.941 \pm 0.079$ & $\underline{0.944 \pm 0.077}$ & $0.936 \pm 0.079$ \\
        \midrule
        \multicolumn{9}{c}{\textbf{Task: \textit{Wavefront estimation (trivial ambiguities do matter).} Metric: \textit{Strehl ratio.}}} \\
        \midrule
        \multirow{2}{*}{Circle} & Noiseless & $0.582 \pm 0.379$ & $\underline{0.605 \pm 0.384}$ & $0.563 \pm 0.372$ & $0.533 \pm 0.261$ & $0.536 \pm 0.270$ & $0.529 \pm 0.278$ & \cellcolor[gray]{0.9}$0.630 \pm 0.274$ \\
        & Noisy & $0.141 \pm 0.070$ & $0.148 \pm 0.076$ & $0.489 \pm 0.285$ & $0.543 \pm 0.225$ & $\underline{0.549 \pm 0.234}$ & \cellcolor[gray]{0.9}$0.551 \pm 0.241$ & $0.545 \pm 0.205$ \\
        \midrule
        \multirow{2}{*}{Triangle} & Noiseless & $0.847 \pm 0.294$ & \cellcolor[gray]{0.9}$1.000 \pm 0.000$ & $0.847 \pm 0.028$ & $0.913 \pm 0.074$ & $0.954 \pm 0.063$ & $0.958 \pm 0.060$ & \underline{$0.976 \pm 0.019$} \\
        & Noisy & $0.155 \pm 0.067$ & $0.157 \pm 0.069$ & $0.664 \pm 0.233$ & \cellcolor[gray]{0.9}$0.837 \pm 0.086$ & $\underline{0.828 \pm 0.091}$ & $0.827 \pm 0.091$ & $0.801 \pm 0.103$ \\
        \midrule
        \multirow{2}{*}{Pentagon} & Noiseless & $0.633 \pm 0.379$ & \cellcolor[gray]{0.9}$0.998 \pm 0.036$ & $0.645 \pm 0.370$ & $0.928 \pm 0.080$ & $0.983 \pm 0.026$ & $0.983 \pm 0.028$ & \underline{$0.985 \pm 0.019$} \\
        & Noisy & $0.136 \pm 0.067$ & $0.145 \pm 0.073$ & $0.508 \pm 0.285$ & \cellcolor[gray]{0.9}$0.779 \pm 0.179$ & $0.767 \pm 0.196$ & $\underline{0.770 \pm 0.190}$ & $0.736 \pm 0.190$ \\
        \midrule
        \multicolumn{9}{c}{\textbf{Task: \textit{Computation speed on $128\times128$ PSF.} Metric: \textit{1 / Time to recover the phase (Hz.)}}} \\
        \midrule
        \multicolumn{2}{c|}{Speed (Hz.)} & $3.33$ & $2.38$ & $0.31$ & \cellcolor[gray]{0.9} $5200$ & \underline{$3700$} & $3300$ & $440$ \\
        \bottomrule
    \end{tabular}
\end{table*}

\subsection{Discussion of simulated results}
In Table~\ref{tab:big_exp} we compare performance across algorithms on phase retrieval and wavefront estimation. Following \cite{Unser_2023_Review}, we use the metric of normalized correlation $\vy^T \widehat{\vy} / \norm{\vy}\norm{\widehat{\vy}}$ for phase retrieval and the Strehl ratio \eqref{eq: strehl} for wavefront estimation. Hence, both metrics vary in the range of $[0, 1]$, with $1$ meaning perfect recovery for the relevant task. We report mean and standard deviation on a testing dataset of 1000 PSFs.

We summarize Table~\ref{tab:big_exp} as follows:
\begin{enumerate}
    \item \textbf{Iterative methods.} Iterative methods perform well on the noiseless phase retrieval task, even when the underlying phase has ambiguity (see circle performance). The $\vxi$-encoded method is best in the presence of noise. Performance on wavefront estimation is mixed; if the objective landscape is highly non-convex than the recovery is suboptimal. The HIO algorithm performs well in the noiseless case for both problems with the exception of circular support (due to ambiguity).
    \item \textbf{Direct methods.} Direct PSF inversion performs well on phase retrieval when the underlying problem is unambiguous; ambiguity appears to slightly impact performance on the phase retrieval problem. Direct methods perform better on wavefront estimation methods compared to their iterative counterparts, especially in the noisy case. The S2P is more robust in noise while the direct PSF alternatives perform better in the noiseless case (likely due to the use of basis representation).
    \item \textbf{Unrolled S2P.} The unrolled S2P tends to perform the best of the supervised methods on the phase retrieval and wavefront estimation tasks in the noiseless case. It is interesting to note that the performance is slightly worse on noisy wavefront estimation against direct PSF inversion, which we hypothesize is due to overparameterization or the iterative mechanism allowing it to (erroneously) fit the noise.
\end{enumerate}

We also report speeds of the various algorithms in Table~\ref{tab:big_exp}, where for iterative methods we have used 500 iterations for the HIO/GS algorithm and 300 iterations for the $\vxi$-encoded INR. The S2P performs approximately $2000\times$ faster than the GS/HIO algorithms and nearly $5\times$ the speed of the high-speed Shack-Hartmann wavefront sensor from Thor labs. In turbulence correction, the speed for optical correction can be on the scale of tens to hundreds of Hertz \cite{Tyson_Frazier_Book}, only the Shack-Hartmann sensor and our methods meet this requirement.

We offer a comparison between supervised and self-supervised training, that is
\begin{equation}
    \argmin{\theta} \norm{\vy - \abs{\mF \left( \vp \odot e^{j g_{\theta}(\vy)} \right) }^2}^2,
    \label{eq: self-sup}
\end{equation}
for the direct mapping $(\vy \to \vphi)$ in Table~\ref{tab:combined_comparison}. The triangle pupil performs best for wavefront estimation, highlighting the concept of \emph{identifiability} in the wavefront estimation problem; we empirically observe this to be related to the degree of pupil asymmetry (similar to iterative algorithms). The improvement for phase retrieval using a circular pupil is likely due to its ability to ignore ambiguity when trained in a self-supervised manner.

\begin{table}
    \centering
    \caption{Comparison of $(\vy \to \vphi)$ performance between Table~\ref{tab:big_exp} and self-supervised approaches (noiseless conditions)}
    \label{tab:combined_comparison}
    
    \begin{tabular}{l|c|c}
        \toprule
        \multicolumn{3}{c}{\textbf{Task: \textit{Phase retrieval.} Metric: \textit{Normalized correlation.}}} \\
        \midrule
        \textbf{Support} & \textbf{Supervised (Table~\ref{tab:big_exp})} & \textbf{Semi-Supervised} \\
        \midrule
        Circle & $0.907 \pm 0.110$ & $0.991 \pm 0.011$ \\
        Triangle & $0.989 \pm 0.028$ & $0.988 \pm 0.020$ \\
        Pentagon & $0.996 \pm 0.010$ & $0.982 \pm 0.031$ \\
        \midrule
        \multicolumn{3}{c}{\textbf{Task: \textit{Wavefront estimation.} Metric: \textit{Strehl ratio.}}} \\
        \midrule
        \textbf{Support} & \textbf{Supervised (Table~\ref{tab:big_exp})} & \textbf{Semi-Supervised} \\
        \midrule
        Circle & $0.529 \pm 0.278$ & $0.563 \pm 0.333$ \\
        Triangle & $0.958 \pm 0.060$ & $0.822 \pm 0.073$ \\
        Pentagon & $0.983 \pm 0.028$ & $0.587 \pm 0.287$ \\
        \bottomrule
    \end{tabular}
\end{table}

\subsection{Real experiments}
A prototype triangular pupil imaging system with a spatial light modulator (SLM) was built as shown in \fref{fig: real-setup}. We use a Holoeye LC 2012 SLM with a 520 nm. laser and record PSFs with a Teledyne Imaging FLIR Grasshopper 3.

Due to limitations in the system/optics we are able to utilize only $22\times22$ pixels for the SLM pattern (see \fref{fig: real-setup} for a visualization of the pupil size) hence our patterns contain only the first 10 Zernike polynomials to avoid aliasing. Due to a slight mismatch between simulated and real data, we train a network on real data only, leaving some for validation. We train a network $g_\theta(\vy) \approx \valpha$ on 135 PSF and phase pairs, and perform validation on the remaining 15. A few examples of real measurements, the phase aberration, and recovery, along with Zernike coefficients are shown in \fref{fig: real-results}.

\begin{figure}
    \centering
    \includegraphics[width=0.9\linewidth]{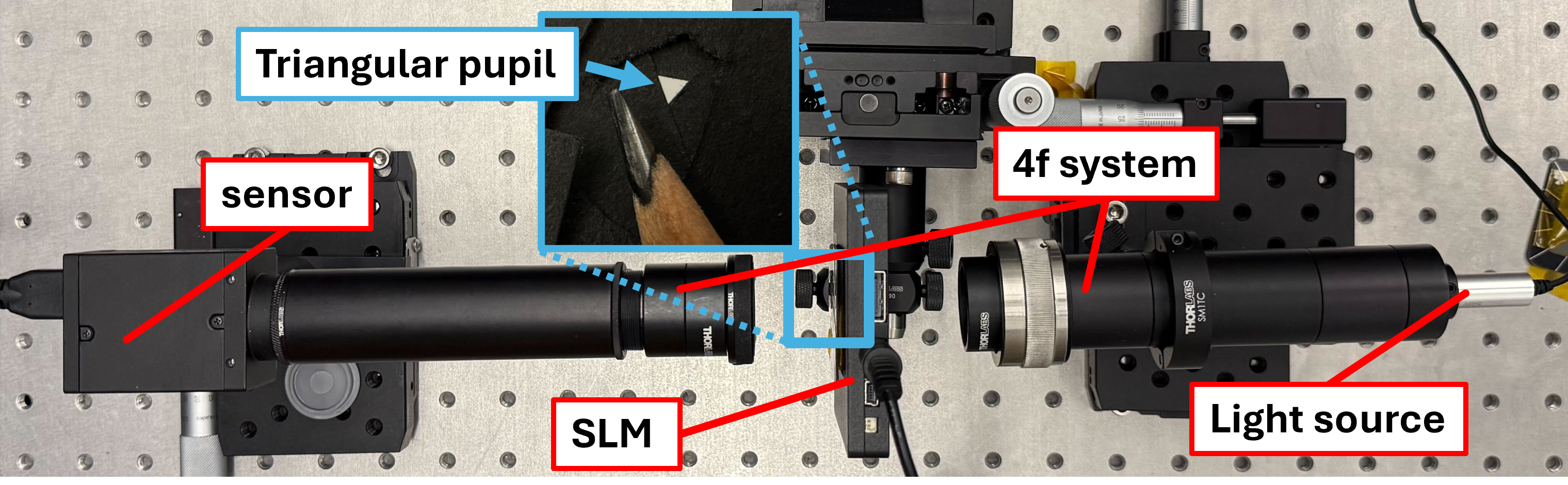} \\
    \includegraphics[width=0.3\linewidth]{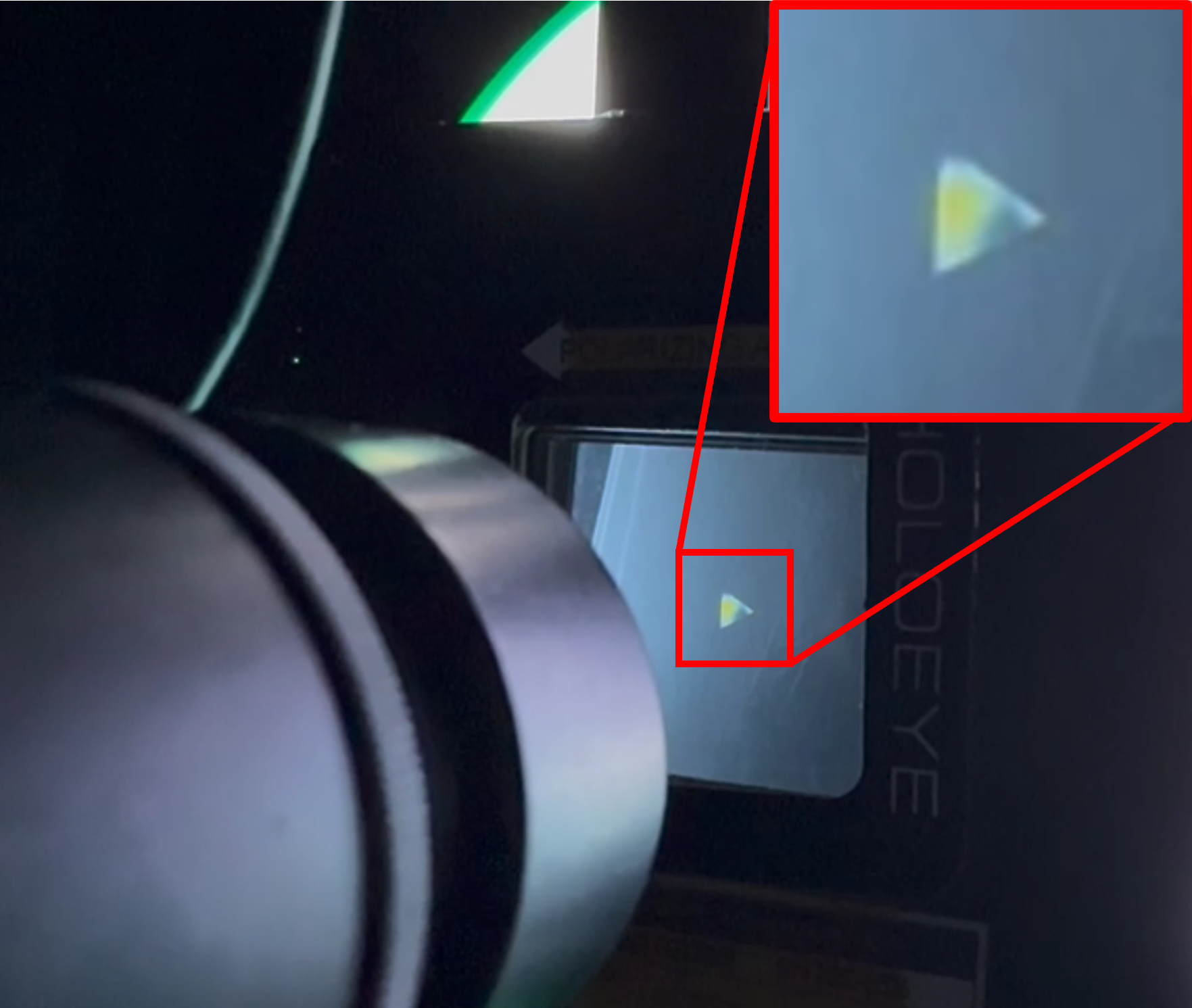}
    \vspace{-1ex}
    \caption{Real experiment setup. [Top] The 4$f$-system, SLM, and triangular pupil setup. [Bottom] Visualization of triangular pattern on SLM.}
    \label{fig: real-setup}
\end{figure}

\begin{figure}
    \centering
    \begin{tabular}{ccc}
    \includegraphics[width=0.28\linewidth]{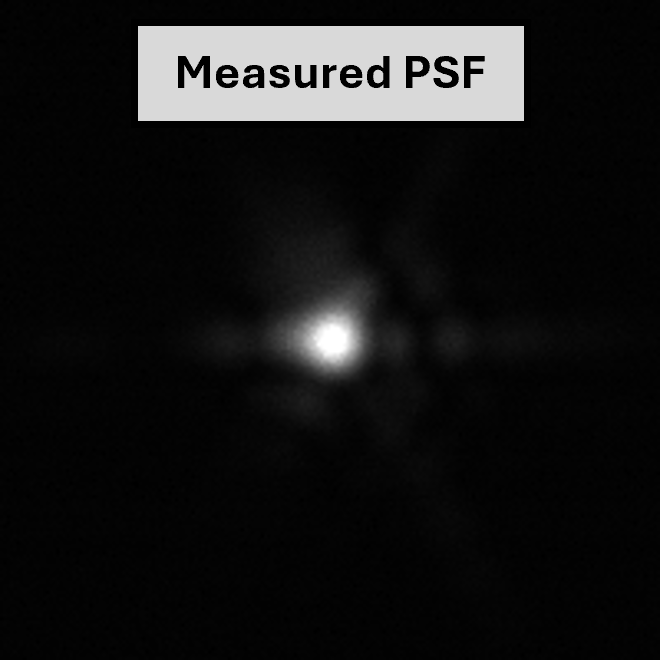} & 
    \includegraphics[width=0.28\linewidth]{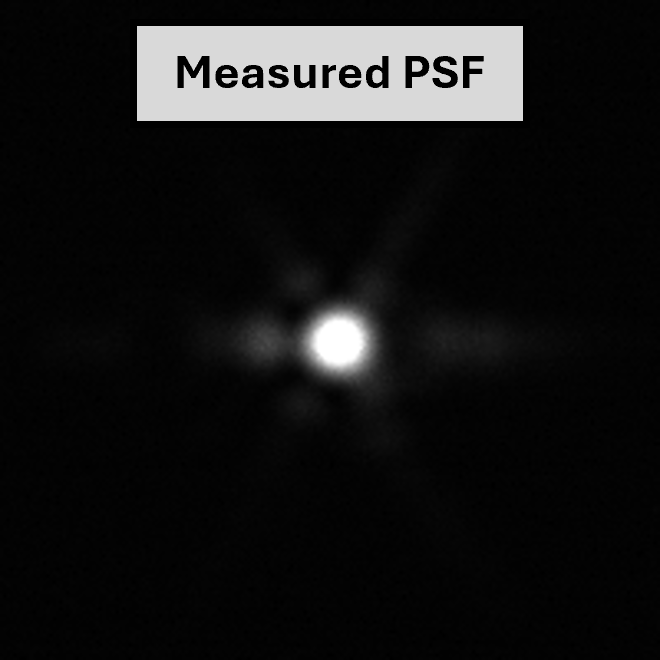} & 
    \includegraphics[width=0.28\linewidth]{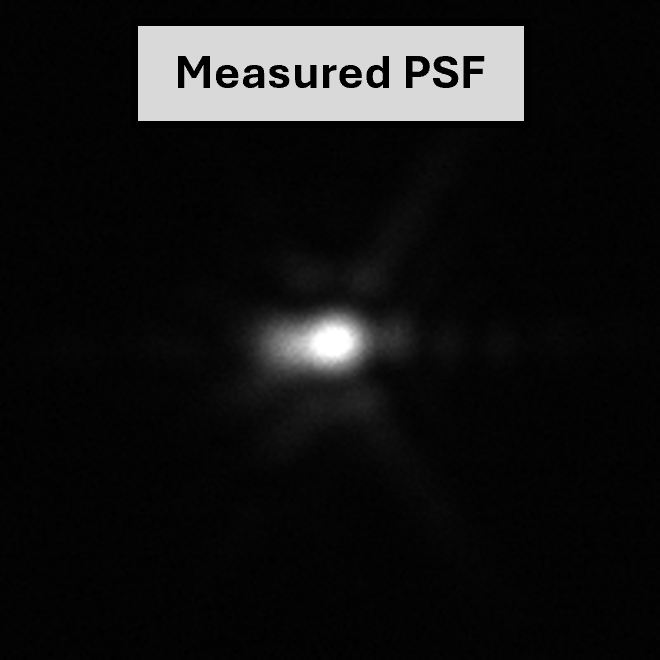} \\
    
    \begin{tabular}{@{}c@{\hspace{1mm}}c@{}}
    \includegraphics[width=0.135\linewidth]{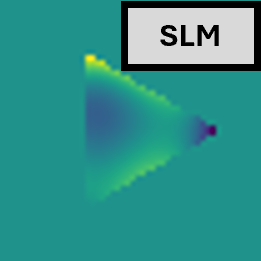} & 
    \includegraphics[width=0.135\linewidth]{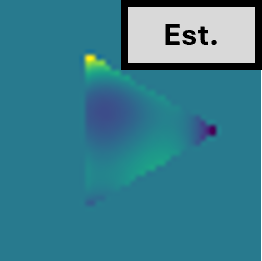}
    \end{tabular} & 
    \begin{tabular}{@{}c@{\hspace{1mm}}c@{}}
    \includegraphics[width=0.135\linewidth]{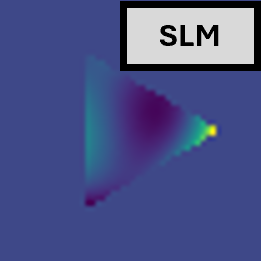} & 
    \includegraphics[width=0.135\linewidth]{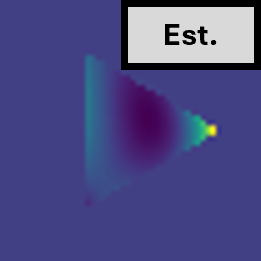}
    \end{tabular} & 
    \begin{tabular}{@{}c@{\hspace{1mm}}c@{}}
    \includegraphics[width=0.135\linewidth]{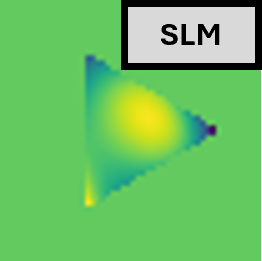} & 
    \includegraphics[width=0.135\linewidth]{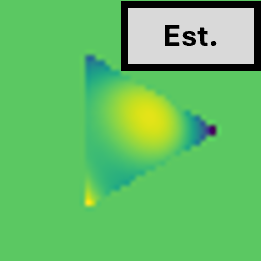}
    \end{tabular} \\
    
    \includegraphics[width=0.28\linewidth]{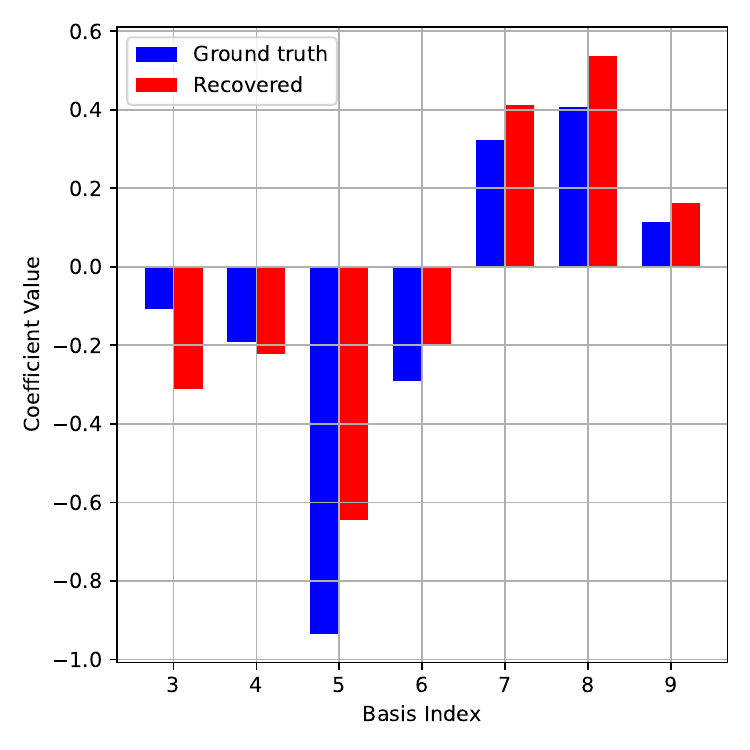} & 
    \includegraphics[width=0.28\linewidth]{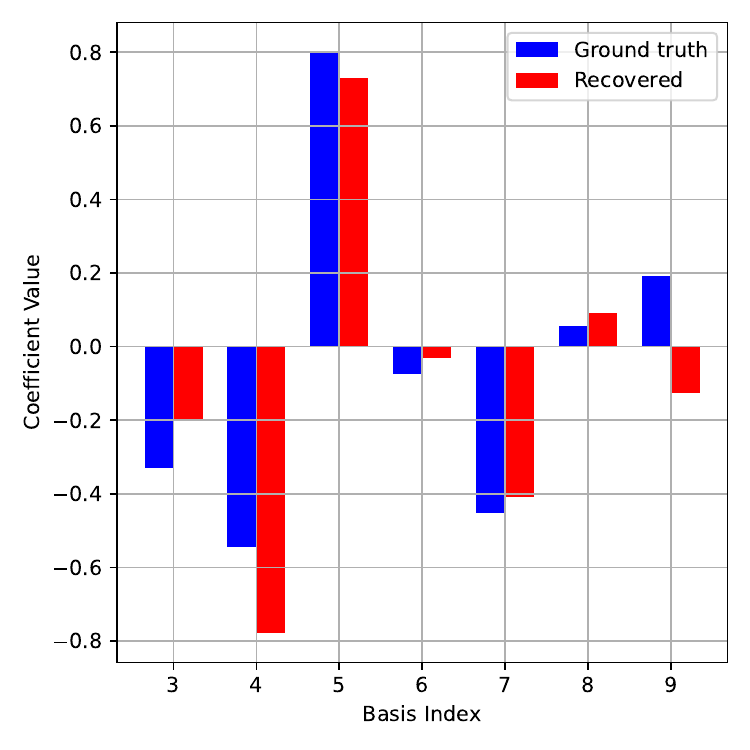} & 
    \includegraphics[width=0.28\linewidth]{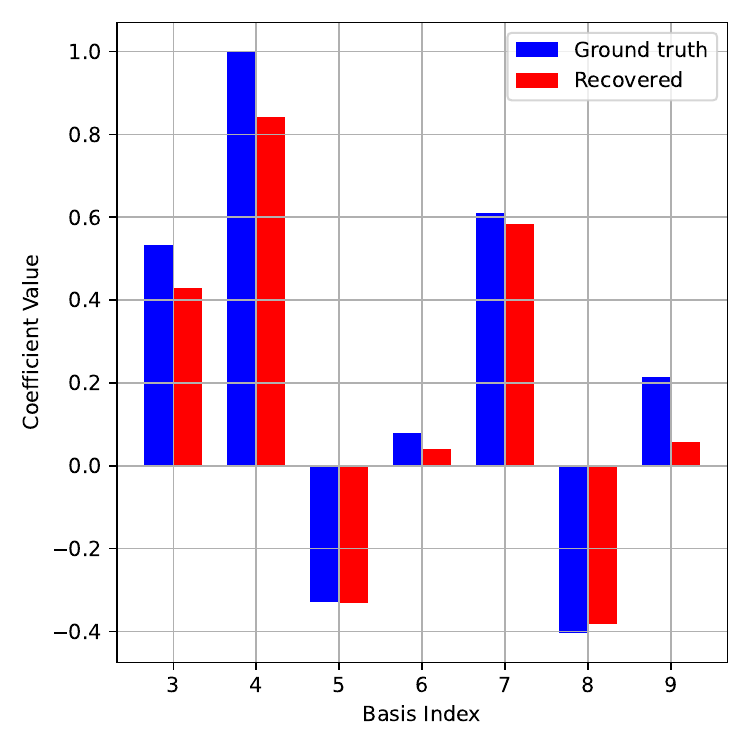} \\
    
    (a) & (b) & (c)
    \end{tabular}
    \caption{Evaluation on \textit{real} validation data. [Top] Measured PSFs. [Middle] SLM phase pattern (left) and estimated phase (right) for each case. [Bottom] Basis coefficient comparison (blue as ground truth, red as recovered). All cases shown (a-c) are validation data. Here we use the MLP-1 method, hence runtime is $\sim3700$ Hz.}
    \label{fig: real-results}
\end{figure}

\section{Conclusion}
In this work we studied the problem of wavefront estimation. We showed that trivial ambiguities play an important role in determining the uniqueness of the solution. We identified a set of conditions under which we can guarantee a unique recovery of the phase. We proposed several algorithms inspired by previous approaches to recover phases from their PSFs in real time and demonstrated our methodology in a real system using an SLM.

With this new approach, we envision that it shall play a role in achieving fast and accurate wavefront estimation in a passive setting, opening the door to a new generation of medical and defense applications. In future works, we look towards answering questions related to robustness and performance bounds as well as optimal pupil geometry and algorithm co-design.

\section*{Acknowledgments}
The authors thank Harshana Weligampola for the helpful discussions and assistance in the optical bench setup.

{\appendices
\section*{Proof of Theorem 2}
\begin{proof}
     Using the Cauchy-Schwarz inequality we may write
     \begin{align*}
         \abs{\vx^H \widehat{\vx}}^2 &\leq (\vx^H \vx) (\widehat{\vx}^H \widehat{\vx}).
     \end{align*}
     With the Cauchy-Schwarz inequality, equality is achieved if and only if $\widehat{\vx}$ is a complex, scalar multiple of $\vx$. This implies
     \begin{equation*}
         \widehat{\vx} = \nu \vx,
     \end{equation*}
     where $\nu \in \C$. Furthermore, since $\widehat{\vx} \in \Omega(\vx)$, $\nu = e^{j c}$, where $c \in [0, 2\pi)$ can be shown from noting the periodicity of the complex exponential. Hence the Theorem follows.
\end{proof}

\section*{Proof of Theorem 3}
\subsection*{Proofs of Lemmas}
We now present three Lemmas that help to remove certain sets of trivial ambiguities. Recall that we have assumed that $\vp$ is known exactly, hence we shall heavily rely upon the knowledge of $\vp$. Specifically, by definition of the support and since it is known \emph{a priori}, we can require that
\begin{equation}
    \vp \odot \widehat{\vx} = \widehat{\vx}
\end{equation}
for any candidate $\widehat{\vx} \in \Omega(\vx)$. This will help to eliminate most of the ambiguities.

Lemma 1 is aimed at removing $\calT_{\text{shift}}$ for $\vell \neq \vzero$.

\begin{lemma}[Aperiodic support]
    Let $\vp$ is known and $\abs{x[\vn]} > 0 \; \forall \vn \in \Z_N$. If $\nexists \vell \in \Z_N$ with $\vell \neq \vzero$ such that $p[\vn - \vell] = p[\vn] \; \forall \vn \in \Z_N$, then $\vp \odot \calT_{\text{shift}}(\vx) \notin \Omega(\vx) \; \forall \vell \neq \vzero \in \Z_N$.
\end{lemma}
\begin{proof}
    Suppose $\vp$ is known and $\nexists\vell \in \Z_N$ for $\vell \neq \vzero$ such that $p[\vn - \vell] = p[\vn]$ and $\calT_{\text{shift}}(\vx) \in \Omega(\vx)$. Since $x[\vn - \vell]$ is an ambiguity by translation, its support too must be a translation. Hence, by definition of $\vp$,
    \begin{equation*}
        p[\vn - \vell] x[\vn - \vell] = x[\vn - \vell], \;\; \forall \vn \in \Z_N.
    \end{equation*}
    Since we know $\vp$, the solution must also satisfy
    \begin{equation*}
        p[\vn] p[\vn - \vell] x[\vn - \vell] = x[\vn - \vell], \;\; \forall \vn \in \Z_N.
    \end{equation*}
    However, this would imply $p[\vn] p[\vn - \vell] = 1 \; \forall \vn \in \Z_N$ for some $\vell \neq \vzero$ since $\abs{x[\vn]} > 0 \; \forall \vn \in \Z_N$ which contradicts our assumption. By contradiction, the Lemma follows.
\end{proof}

Lemma 2 is aimed at eliminating $\calT_{\text{shift}} \circ \calT_{\text{flip}}$. Note that this also eliminates $\calT_{\text{flip}}$ if we consider a shift of $\vell = \vzero$.
\begin{lemma}[Asymmetric and aperiodic support]
    Let $\vp$ be known and $\abs{x[\vn]} > 0 \; \forall \vn \in \Z_N$. Then $\exists \vell \in \Z_N$ such that $\vp \odot (\calT_{\text{shift}} \circ \calT_{\text{flip}})(\vx) \in \Omega(\vx)$ if and only if $(\calT_{\text{shift}} \circ \calT_{\text{flip}})(\vp) = \vp$ for the same $\vell$.
\end{lemma}
\begin{proof}
    ($\implies$). Suppose $\exists \vell \in \Z_N$ such that $(\calT_{\text{shift}} \circ \calT_{\text{flip}})(\vx) \in \Omega(\vx)$. Then
    \begin{equation*}
        p[-\vn - \vell] x^*[-\vn - \vell] = x^*[-\vn - \vell].
    \end{equation*}
    Since $\vp$ is known, this also implies
    \begin{equation*}
        p[\vn] p[-\vn - \vell] x^*[-\vn - \vell] = x^*[-\vn - \vell]
    \end{equation*}
    which further implies $p[\vn] p[-\vn - \vell] = 1 \; \forall \vn \in \Z_N$ since $\abs{x[\vn]} > 0 \; \forall \vn \in \Z_N$. By definition of $\vp$, this implies $p[\vn] = p[-\vn - \vell]$ for all $\vn \in \Z_N$.

    ($\impliedby$). Suppose $\exists \vell \in \Z_N$ such that $(\calT_{\text{shift}} \circ \calT_{\text{flip}})(\vp) = \vp$. Consider a transformation of the true solution $\vx$: $(\calT_{\text{shift}} \circ \calT_{\text{flip}})(\vx)$. Then the following holds $\forall \vn \in \Z_N$:
    \begin{align*}
        p[-\vn - \vell] x^*[-\vn - \vell] &= x^*[-\vn - \vell] \\
        \implies p[\vn] x^*[-\vn - \vell] &= x^*[-\vn - \vell].
    \end{align*}
    Hence, $\vp \odot (\calT_{\text{shift}} \circ \calT_{\text{flip}})(\vx) \in \Omega(\vx)$.

    ($\iff$). Follows from both directions.
\end{proof}

The final Lemma is aimed at eliminating $\calT_{\text{dc}}$.
\begin{lemma}[Known complex DC shift]
    Let $\abs{x[\vn]} > 0 \; \forall \vn \in \Z_N$. If $\vphi^T \vp / \norm{\vp} = c_{0}$ with $c_{0} \in [0, 2\pi)$ $\implies$ $\nexists$ $c' \in (0, 2\pi)$ such that $\vx = e^{jc'} \vx$.
\end{lemma}
\begin{proof}
    We first expand
    \begin{equation*}
        \phi[\vn] = c_{0} + \Tilde{\phi}[\vn], \; \forall \vn \in \Z_N
    \end{equation*}
    and let $\vp^T \vphi  / \norm{\vp} = c_{0}$ and $\vp^T \Tilde{\vphi} / \norm{\vp} = 0$. Then $\forall \vn \in \Z_N$,
    \begin{align*}
        x[\vn] &= e^{j c_{0}} \abs{x[\vn]} e^{j \Tilde{\phi}[\vn]}.
    \end{align*}
    Now suppose $\exists$ $c' \in (0, 2\pi)$ such that $\vx = e^{jc'} \vx$. Since $\vert\vx\vert > \vzero$, this implies $e^{j(c_{0} + c')} = e^{j c_{0}}$ which further implies
    \begin{equation*}
        e^{jc'} = 1, \; c' \in (0, 2\pi).
    \end{equation*}
    This can not be true, hence the Lemma is shown.
\end{proof}

\subsection*{Proof of Theorem 3}

\begin{proof}
Define $\vy \in \R^{K}$ as vector according to measurement model \eqref{eq: measurement_model} and further suppose that there exists no non-trivial ambiguity, i.e., $\nexists \widehat{\vx} \notin \Omega(\vx)$ such that $\vy = \abs{\mF \widehat{\vx}}^2$. Further suppose we meet the conditions specified of Theorem~\ref{th: theorem_iPSF}. Note that we have not yet asserted Hayes' Theorem.

Recall that the set of trivial ambiguities relative to true solution $\vx$ is as follows:
\begin{equation*}
    \Omega(\vx) = \langle \{ \calT_{\text{dc}}(\vx), \; \calT_{\text{shift}}(\vx), \; \calT_{\text{flip}}(\vx) \} \rangle.
\end{equation*}
By applying Lemma 3, which provides uniqueness relative to $\calT_{\text{dc}}$ applied through condition (iii), $\Omega(\vx)$ can be shown to reduce to
\begin{align*}
    \Omega(\vx) &= \langle \{ \calT_{\text{shift}}(\vx), \; \calT_{\text{flip}}(\vx) \} \rangle \\
    &= \{ \calT_{\text{shift}}(\vx), \; \calT_{\text{flip}}(\vx), \; (\calT_{\text{shift}} \circ \calT_{\text{flip}})(\vx) \}.
\end{align*}
By Lemma 2, which eliminates ambiguities related to the flipped solution facilitated by conditions (i) and (ii), this further can be shown to reduce to
\begin{equation*}
    \Omega(\vx) = \{ \calT_{\text{shift}}(\vx) \}.
\end{equation*}
By Lemma 1, related to the shift ambiguity and applied via condition (i) and (ii), $\calT_{\text{shift}}$ is eliminated for all $\vell \in \Z_N$ such that $\vell \neq \vzero$. Therefore, the application of Lemma 1 leads to
\begin{equation*}
    \Omega(\vx) = \{ \vx \},
\end{equation*}
i.e., the only remaining solution is indeed $\vx$.

Hence, if no non-trivial ambiguities exist, under conditions stated in Theorem~\ref{th: theorem_iPSF} the solution is unique. To arrive at Theorem~\ref{th: theorem_iPSF}, one need only require $\vy$ to satisfy Hayes' Theorem (and accordingly a non-trivial ambiguity occurs with probability 0), and the Theorem follows.
\end{proof}
}

\section*{Proof of Theorem 4}
\begin{proof}
    Let $\psi[\vn]$ be known over $c[\vn]$ and $y[\vn]$ satisfy Hayes' Theorem, hence, all ambiguous solutions (with probability 1) belong to the set of trivial ambiguities $\Omega(\vx)$ where $x[\vn] = p[\vn] e^{j \psi[\vn]}$ is the true solution.
    
    Suppose there exists another support $\widehat{\vp} \neq \vp$ such that $\widehat{p}[\vn] e^{j \psi[\vn]} \in \Omega(\vx)$. It can be observed by the definition of $\Omega(\vx)$ that the magnitudes of all solutions that belong to $\Omega(\vx)$ are of the form of translations and flips of $\vp$. Therefore, $\widehat{p}[\vn] = p[\pm \vn + \vell]$ for some $\vell \in \Z_N$. Since $\vpsi$ is also known, this implies $p[\pm \vn + \vell] e^{j \psi[\vn]} \in \Omega(\vx)$.

    One can then observe that there does exist $\vell \in \Z_N$ such that $\widehat{\vp} \neq \vp$ and $p[\pm \vn + \vell] e^{j \psi[\vn]} \in \Omega(\vx)$. This represents the contrapositive case, thus the Theorem follows.
\end{proof}

\bibliographystyle{IEEEbib_abv}
\bibliography{refs}

\end{document}

%% file: commands.tex
\usepackage[table]{xcolor}

\usepackage{amsmath}
\usepackage{amsbsy}
\usepackage{amssymb}
\usepackage{amsfonts}
\usepackage{amsthm}
\usepackage{physics}
\usepackage{datetime}
\usepackage{graphicx}
\usepackage{algorithmic}
\usepackage{algorithm}
\usepackage{tcolorbox}
\usepackage{enumerate}

\providecommand{\eref}[1]{\eqref{#1}}  
\providecommand{\cref}[1]{Chapter~\ref{#1}}

\providecommand{\fref}[1]{Figure~\ref{#1}}

\renewcommand{\stop}{\text{S2P}}

\newtheorem{property}{Property}
\newtheorem{corollary}{Corollary}

\newtheorem{definition}{Definition}
\newtheorem{theorem}{Theorem}
\newtheorem{lemma}{Lemma}

\providecommand{\R}{\ensuremath{\mathbb{R}}}

\providecommand{\C}{\ensuremath{\mathbb{C}}}

\providecommand{\Z}{\ensuremath{\mathbb{Z}}}

\providecommand{\abs}[1]{\lvert#1\rvert}
\providecommand{\norm}[1]{\lVert#1\rVert}

\providecommand{\bydef}{\overset{\text{def}}{=}}

\renewcommand{\vec}[1]{\ensuremath{\mathbf{#1}}}
\providecommand{\greekvec}[1]{\ensuremath{\boldsymbol{#1}}}
\providecommand{\mat}[1]{\ensuremath{\mathbf{#1}}}

\providecommand{\calN}{\mathcal{N}}

\providecommand{\calT}{\mathcal{T}}

\providecommand{\mF}{\mat{F}}

\providecommand{\mI}{\mat{I}}

\providecommand{\vm}{\vec{m}}
\providecommand{\vn}{\vec{n}}

\providecommand{\vp}{\vec{p}}

\providecommand{\vx}{\vec{x}}
\providecommand{\vy}{\vec{y}}

\providecommand{\vell}{\boldsymbol{\ell}}

\providecommand{\valpha}{\greekvec{\alpha}}
\providecommand{\vbeta}{\greekvec{\beta}}

\providecommand{\vxi}{\greekvec{\xi}}

\providecommand{\vphi}{\greekvec{\phi}}

\providecommand{\vpsi}{\greekvec{\psi}}

\tcbuselibrary{theorems}
\newtcbtheorem[]{theo}{Theorem}%
{colback=gray!30,colframe=gray!45!black,fonttitle=\bfseries}{th}
\newtcbtheorem[]{asmp}{Assumption}%
{colback=gray!30,colframe=gray!45!black,fonttitle=\bfseries}{th}
\newtcbtheorem[]{prop}{Property}%
{colback=gray!30,colframe=gray!45!black,fonttitle=\bfseries}{th}
\newtcbtheorem[]{corol}{Corollary}%
{colback=gray!30,colframe=gray!45!black,fonttitle=\bfseries}{th}

\providecommand{\vxtilde}{\mathbf{\widetilde{x}}}

\providecommand{\vzero}{\vec{0}}

\newcommand{\argmin}[1]{\mathop{\underset{#1}{\mbox{argmin}}}}

